\documentclass[11pt,letterpaper]{article}
\usepackage{graphicx}

\usepackage{hyperref}

\usepackage[table]{xcolor}
\usepackage{arydshln}
\usepackage{amsmath,amsfonts, amssymb, amsthm}
\usepackage{commath}
\usepackage{dsfont} 
\usepackage{bm} 
\usepackage{tabularx}

\usepackage{subcaption}
\usepackage{tikz}
\usetikzlibrary{backgrounds}
\usetikzlibrary{calc}	
\usetikzlibrary{patterns}
\usetikzlibrary{positioning}	
\usetikzlibrary{arrows.meta}
\usetikzlibrary{shapes}
\usepackage{setspace}

\usepackage[margin=1in]{geometry}
\usepackage{xspace}
\usepackage{enumitem}

\usepackage{url}

\usepackage{color}
\usepackage{xcolor}
\usepackage{esvect}[b]

\usepackage{thmtools} 
\usepackage{thm-restate} 

\theoremstyle{definition}
\newtheorem{theorem}{Theorem}
\newtheorem*{theorem*}{Theorem}

\newtheorem{claim}[theorem]{Claim}
\newtheorem*{claim*}{Claim}
\newtheorem{lemma}[theorem]{Lemma}
\newtheorem{corollary}[theorem]{Corollary}

\newcommand{\reals}{\mathbb{R}}
\newcommand{\cost}{\textnormal{cost}}
\newcommand{\dist}{\textnormal{dist}}
\newcommand{\OPT}{\textnormal{OPT}}
\newcommand{\opt}{\textnormal{opt}}
\newcommand{\ALG}{\textnormal{ALG}}
\newcommand{\ADV}{\textnormal{ADV}}

\newcommand{\rooot}{\textnormal{root}}

\newcommand{\OPTOTHER}{{\textnormal{OPT\textsuperscript{other}}}}

\usepackage{stmaryrd}
\usepackage{trimclip}

\newcommand{\tilT}{\widetilde T}
\newcommand{\TT}{T}

\usepackage{tikz}
\usetikzlibrary {arrows.meta,bending,positioning}

\title{Chasing Small Sets Optimally Against Adaptive Adversaries}

\author{Christian Coester\\ University of Oxford \and Alexa Tudose\\ University of Oxford}

\date{}

\begin{document}

\maketitle

\begin{abstract}
We study deterministic online algorithms for the problem of chasing sets of cardinality at most $k$ in a metric space, also known as metrical service systems and equivalent to width-$k$ layered graph traversal. We resolve the 30-year-old gap of $\Omega(2^k)\cap O(k2^k)$ on the competitive ratio of this problem by giving an $O(2^k)$-competitive deterministic algorithm. This bound is optimal even among randomized algorithms against adaptive adversaries. We also (slightly) improve the deterministic lower bound to $D_k$, defined recursively by $D_1=1$ and $D_{k+1}=2D_k+\sqrt{8+8D_k}+3$, which we conjecture to be exactly tight. For $k=3$, we provide a matching upper bound of $D_3$. Our results imply slightly improved upper and lower bounds for distributed asynchronous collective tree exploration and for the $k$-taxi problem, respectively.

Our algorithm generalizes the classical doubling strategy, previously known to be optimal for $k=2$. The previous best bound for general $k$ was achieved by the generalized work function algorithm (WFA), and was known to be tight for WFA. Our improved bound therefore implies that WFA is sub-optimal for chasing small sets.
\end{abstract}

\pagenumbering{gobble}
\newpage
\pagenumbering{arabic}

\section{Introduction}

\paragraph{Problem definition.} Small set chasing, also known as metrical service systems, was introduced by Chrobak and Larmore in 1991 \cite{ChLar-ServerProblem}. In this problem, a player must move in a metric space $(M, d)$ in order to serve a sequence of requests. The player is initially located at $s_0 \in M$. At time $t\in\{1, \dots, T\}$, a set $S_t \subseteq M$ with $1 \leq \vert S_t \vert \leq k$ is revealed, and the player has to relocate to a point $s_t \in S_t$, paying cost $d(s_{t-1}, s_t)$. The number of requests $T$ and the parameter $k$ are not known to the player in advance. The goal of the player is to minimize the total incurred cost. The player is said to be $\rho$-competitive if its cost is at most 
\[
\rho \cdot \min \left\{ \sum_{t=1}^T d(x_{t-1}, x_t) \,\bigg\vert\, x_0 = s_0, x_1 \in S_1, \dots, x_T \in S_T \right\}.
\]

\paragraph{Related work on deterministic algorithms.} When introducing the problem, Chrobak and Larmore gave an optimal $k$-competitive deterministic algorithm for uniform metric spaces, and a $9$-competitive deterministic algorithm for $k=2$ in arbitrary metric spaces \cite{ChLar-ServerProblem}. The first competitive deterministic algorithm for general metric spaces and arbitrary $k$ is due to \cite{Fiat-competitiveLGT} and achieves a competitive ratio of $O(9^k)$. In the same paper, they also show a lower bound of $2^{k-2}$ for deterministic algorithms. By designing better algorithms, the initial exponential gap of $O((9/2)^k)$ between upper and lower bounds was narrowed to polynomial $O(k^3)$ by Ramesh \cite{Ramesh-LGT}, and further to linear $O(k)$ by Burley \cite{burleyLGT}, at which the problem had been stuck for 30 years. We close the gap to a constant factor, with a different algorithm than previous works. 

The algorithm proposed by Burley is the generalized work function algorithm (WFA), a general purpose algorithm applicable to many online problems, and the leading candidate algorithm for resolving the famous $k$-server conjecture: For the $k$-server problem, WFA is known to achieve the optimal competitive ratio up to a factor of at most $2$ \cite{kServerConjecture}, and it is conjectured to be exactly optimal. By contrast, our results imply that WFA is not optimal for chasing small sets, since Burley showed that his analysis is tight for WFA~\cite{burleyLGT}.

\paragraph{Context and motivation.} The central role of chasing small sets is emphasized by its connection to numerous other problems. Its alternative name, metrical service systems, showcases the connection to metrical task systems, a more general problem in which each request comes with a function that specifies the cost to serve the request in each point of the metric space \cite{MTSintroduction}. Note that metrical service systems is a particular case where the service costs are in $\{0, \infty\}$. Another related problem is convex body chasing \cite{Friedman1993,BubeckLLS23,Sellke20,ArgueGTG21}, where requested sets are not subject to cardinality constraints, but are required to be convex.

Furthermore, small set chasing is equivalent to the layered graph traversal problem, introduced by Papadimitriou and Yannakakis in \cite{shortestPathsWithoutAMap}. In fact, some of the earlier mentioned results on small set chasing were obtained by working on layered graph traversal. In this problem, a searcher needs to traverse a connected undirected graph $G$ with non-negative edge weights. The vertices of $G$ are partitioned into layers $L_0, \dots, L_N$, so that edges exist only between vertices in consecutive layers. The searcher is initially located in a starting vertex $s$ in $L_0$ and needs to reach a target vertex $t$. Initially, the searcher knows only $s$, the vertices in $L_1$, and the weighted edges connecting $s$ to these vertices. All other vertices and edges are hidden. When the searcher reaches a vertex in $L_i$ for the first time, the vertices in $L_{i+1}$ are revealed, together with the weighted edges between $L_i$ and $L_{i+1}$.  Then, the searcher needs to move from its current position to a vertex in $L_{i+1}$, using any edges which have been revealed so far.\footnote{This may require back-tracking through previous layers.} When the searcher reaches the penultimate layer, the last layer containing only the target vertex $t$ is revealed, together with the corresponding edges. If the graph $G$ contains at most $k$ vertices in each layer, we say that $G$ has width $k$. Note that the searcher does not know $k$ or the number of layers in advance. The goal of the searcher is to minimize the total distance traveled. The searcher is said to be $\rho$-competitive if its total traveled distance is at most $\rho \cdot \dist(s, t)$, where  $\dist(s, t)$ is the distance between $s$ and $t$ in $G$. 

Computing the offline optimum in layered graph traversal simultaneously captures the shortest path problem and the dynamic programming paradigm, and therefore gives a natural motivation for the problem. The special case in which the graph is unweighted has been studied in \cite{UnweightedLGT}, but we focus on the general, weighted version of the problem. In the equivalence between small set chasing and layered graph traversal, requests correspond to layers, and the width of the layered graph corresponds to the maximum cardinality of a set in a request. A proof of this equivalence can be found in \cite{Fiat-competitiveLGT}. 

For width $k=2$, layered graph traversal is equivalent to the linear search problem (also known as the cow path problem), studied since at least 1970~\cite{BeckN70} and later popularized in the online algorithms literature following~\cite{Baeza-YatesCR93}. In this problem, a searcher located on a line tries to find an unknown target. The optimal deterministic competitive ratio of 9 is achieved by a simple doubling strategy: the searcher explores the two directions alternately up to a threshold that increases in powers of 2. Similar doubling strategies are ubiquitous in online algorithms (e.g.,~\cite{Anand21}). Our algorithm for general $k$ can be seen as a recursive generalization of this simple idea.

Layered graph traversal can be reduced to the $k$-taxi problem, where one has to move a set of $k$ servers (i.e., taxis) in a metric space in order to serve incoming passenger requests consisting of a starting point and a destination \cite{k-taxi}. Moreover, the problem of distributed asynchronous collective tree exploration, where a tree is explored by a team of $k$ agents, reduces to layered graph traversal \cite{CossonM25}. Additionally, small set chasing also has applications in designing learning-augmented algorithms: \cite{onlineMetricAlgo} combine multiple online algorithms by essentially reducing this task to layered graph traversal on disjoint paths, and \cite{mixing-predictions} show how to obtain a hybrid algorithm that is competitive against the best dynamic combination of the available predictors by reducing\footnote{The general problem considered in \cite{mixing-predictions} is in fact equivalent to layered graph traversal; the straightforward reduction in the reverse direction constructs $k$ predictors whose suggested positions at any time step cover the entire current layer.} the problem to layered graph traversal. More generally, set chasing may be viewed as a \emph{meta problem} of online decision making, where the currently requested set corresponds to the legal configurations that an algorithm may choose at that step.

\paragraph{Randomized algorithms and different adversarial models.} In the setting of randomized algorithms, an algorithm's cost is defined by taking the expectation over its randomness. Most work on randomized online algorithms focuses on the \emph{oblivious adversary} model, as this is the only model where randomization can yield exponential improvements~\cite{Ben-DavidBKTW94}.
An oblivious adversary is one that fixes an instance upfront, before any random choices by the algorithm are made. For this setting, Chrobak and Larmore \cite{ChLar-ServerProblem} gave an optimal $H_k$-competitive randomized algorithm for chasing small sets on uniform metric spaces and a $4.6$-competitive randomized algorithm for $k=2$ on general metric spaces. For arbitrary $k$ and general metric spaces, the first polynomial competitive ratio of $O(k^{13})$ is due to Ramesh \cite{Ramesh-LGT}, who also gave a lower bound of $\Omega\left(\frac{k^2}{\log^{1+\epsilon}{k}}\right)$ for arbitrary $\epsilon > 0$. Recently, these bounds were tightened to $\Theta(k^2)$ by Bubeck, Coester and Rabani\ \cite{randomizedLGT,BubeckCR23}. Interestingly, ~\cite{BubeckCR23} note that lower bounds for layered graph traversal inspired their construction used to refute the randomized $k$-server conjecture.

By contrast, in \emph{adaptive adversary} models, requests are allowed to depend on past decisions made by the algorithm. It is well known that randomization offers only limited benefit against such adversaries.\footnote{In general, a $\rho$-competitive randomized algorithm against adaptive online adversaries implies a $\rho^2$-competitive deterministic algorithm~\cite{Ben-DavidBKTW94}.} For small set chasing under adaptive adversaries, we show that randomization can improve the competitive ratio by at most a constant factor. This model is particularly relevant to chasing small sets, as several close connections to other problems are known to hold only in the adaptive setting. For example, the previously-mentioned reduction from small set chasing to $k$-taxi given in \cite{k-taxi} requires adaptive adversaries, as it issues new requests depending on previous algorithmic actions. Similarly, the reduction from distributed asynchronous collective tree exploration to layered graph traversal in \cite{CossonM25} requires an adaptive adversary. Here,  the request set in small set chasing corresponds roughly to (part of) the boundary of the region explored so far, which itself depends on prior moves of the agents. The best known bounds for both problems are improved as a consequence of our results.

\subsection{Our results}
Our main result is the following.
\begin{theorem} \label{th:upper-bound}
    There exists an $O(2^k)$-competitive deterministic algorithm for chasing sets of cardinality at most $k$.
\end{theorem}

Our upper bound matches the asymptotic lower bound of $\Omega(2^k)$, and thus settles the optimal asymptotic competitive ratio of small set chasing (and, equivalently, layered graph traversal). In fact, we show in  Appendix \ref{sec:LB-adaptive} that our algorithm is asymptotically optimal not only among deterministic algorithms, but also among randomized algorithms competing against adaptive adversaries. 

\begin{theorem} \label{th:adaptive-adversaries}
    Every randomized online algorithm for chasing sets of cardinality at most $k$ has competitive ratio at least $2^{k-1}$ against an adaptive online adversary.
\end{theorem}

We additionally improve the exact lower bound for deterministic algorithms from $2^{k-2}$~\cite{Fiat-competitiveLGT} to $D_k$, defined recursively by
\begin{align}
    D_1=1 \quad \text{and} \quad D_{k+1}=2D_k+\sqrt{8+8D_k}+3.\label{eq:DkViaDk-1}
\end{align}

\begin{theorem}
    Every deterministic online algorithm for chasing sets of cardinality at most $k$ has a competitive ratio of at least $D_k$, for $D_k$ as defined in~\eqref{eq:DkViaDk-1}.
    \label{th:lower-bound}
\end{theorem}

Note that $D_2 = 9$ is known to be the optimal competitive ratio for $k=2$. We conjecture that $D_k$ is exactly tight for all $k \geq 1$. We prove that this indeed holds for $k=3$.

\begin{theorem} \label{th:upper-bound-D3}
    There exists a $D_3$-competitive deterministic online algorithm for chasing sets of cardinality at most $3$.
\end{theorem}

Since distributed asynchronous collective tree exploration reduces to width-$k$ layered graph traversal \cite{CossonM25}, and width-$k$ layered graph traversal reduces to $k$-taxi \cite{k-taxi}, we also obtain improved upper and lower bounds, respectively, for these problems.

\begin{corollary} 
    \label{cor:tree-exploration}
    There exists a distributed asynchronous algorithm
that explores any tree of $n$ nodes and depth $D$ in at most $2n + O(k 2^k D)$ moves.
\end{corollary}

\begin{corollary} 
    \label{cor:lower-bound-ktaxi}
    Every deterministic online algorithm for $k$-taxi has a competitive ratio of at least $D_k$, for $D_k$ as defined in~\eqref{eq:DkViaDk-1}.
\end{corollary}

Since competitive $k$-taxi algorithms for general (infinite) metric spaces are unknown except for $k\le 3$~\cite{CoesterP26}, the constant-factor improvement is admittedly modest. Still, we hope it may offer insight for future algorithm design, especially since $D_2=9$ is known to be tight for the $2$-taxi problem~\cite{k-taxi}. In a similar vein, our algorithm for small set chasing is directly inspired by our lower bound construction.

\subsection{Overview of techniques}
Layered graph traversal is known to be equivalent to the case where the graph is a tree~\cite{Fiat-competitiveLGT}. To simplify the description and analysis of our algorithms, it is convenient to further reduce the problem to a two-player game on a tree which evolves over time, as already done in \cite{randomizedLGT}. This reduction is very natural: after layer $L_i$ is revealed, the algorithm only needs to remember the Steiner tree which connects the starting vertex $s$ to the vertices in $L_i$ (all vertices which are not part of this Steiner tree are ``dead-ends'', since they do not have any descendant in $L_i$). Moreover, moving to a vertex in $L_i$ is equivalent to moving to a leaf of the Steiner tree. Therefore, we can model layered graph traversal as ``chasing'' the leaves of the Steiner tree, where the Steiner tree changes when a new layer is revealed. The game formulation helps express the changes of the Steiner tree as a sequence of simple transformations, which we can reason about more easily.

In Section~\ref{sec:LB}, we first provide a sketch of our lower bound of $D_k$ for deterministic algorithms, which will serve to motivate our algorithm. As in the $2^{k-2}$ lower bound of~\cite{Fiat-competitiveLGT}, we use a recursive construction consisting of two branches, each containing a concatenation of lower bound instances of width $k-1$. In each step, the branch currently occupied by the algorithm is extended with an additional lower bound instance, while the other branch advances with zero-length edges. Unlike the construction of~\cite{Fiat-competitiveLGT}, which could be stopped at any time to yield a ratio of $2^{k-2}$, our improved lower bound is only attained at certain times when the ratio between the branch lengths is maximized over the course of the algorithm, and the bound we obtain depends on this ratio. Since we can stop an instance only at those specific moments, the adversary cannot easily control the absolute cost of an instance, as it depends on the behavior of the algorithm. Establishing a rigorous lower bound therefore requires some care: On the one hand, the cost of a recursive instance should not be too large, as otherwise an online algorithm could avoid it by switching to the other branch. On the other hand, it also should not be too small (e.g., exponentially decaying), as we need to be able to stack many recursive instances to obtain unbounded total cost. In details deferred to Appendix~\ref{app:LB}, we show that we can lower and upper bound the cost of an instance, which allows to use them effectively in the recursive construction.

Our lower bound suggests that there is a ``sweet spot'' for the ratio between the two branch lengths at which point the algorithm should switch to the other branch. Specifically, for instances of width $k$, the algorithm should switch to the other branch if the optimal value in its own branch is a factor $x_k := 1 + \sqrt{\frac{2}{1 + D_{k-1}}}$
larger than the optimal value in the other branch. This immediately suggests a simple online algorithm: switch to the shorter branch when the sweet spot of ratio $x_k$ between branch lengths is reached, and employ an analogous strategy (for smaller $k$) recursively within each branch. In Appendix~\ref{app:D3}, we will show that this strategy is indeed $D_3$-competitive for instances of width $3$.

However, generalizing the potential used for width $3$ to larger $k$ does not seem straightforward, as outlined in Appendix~\ref{app:towards-D4}. 
In Section~\ref{sec:almost}, we derive a natural potential function inspired by the lower bound and show that it can \emph{almost} prove $D_k$-competitiveness, except for two types of problems triggered by dead-ends in layered graph traversal. We then address these issues in Section~\ref{sec:mainAlgo} by refining the algorithm and potential function through two methods, which we call \emph{forgetting} and \emph{imbalancing}. While the presentation of these techniques will be specific to layered graph traversal, the underlying principles may be applicable more broadly. The motivating idea is to transform an instance into a canonical hard form, as witnessed in the lower bound construction. \emph{Forgetting} corresponds to treating parts of the instance as if they had not yet been revealed, which we execute by truncating terms in the potential function. \emph{Imbalancing} is a method of rounding an instance towards the highly unbalanced structure of worst-case instances, by distorting some edge lengths. Mathematically, the benefit of such distortions manifests through larger potential function values, and only costs a constant factor in the competitive ratio.

\subsection{Evolving tree game}
\label{sec:evolving-tree-game}
We now describe the evolving tree game, which is essentially the same as the one defined in \cite{randomizedLGT} except that our growth operation is discrete rather than continuous and we define the game directly for binary trees.

We say that a tree is a \textit{stemmed binary tree} if it satisfies the following conditions: a) it is a rooted tree, and the root has degree $1$; b) all vertices except the root have either $0$ or $2$ children. 

The evolving tree game is a two-player game involving stemmed binary tree $T$ with nonnegative edge weights. We denote by $w_e$ the weight of an edge $e$. For a non-root node $u$, we denote by $e(u)$ the edge between $u$ and its parent. Initially, the tree consists of two vertices connected by a zero-length edge. The first player, called \textit{the adversary}, can apply the following operations on $T$:
\begin{itemize}
    \item \textbf{Growth}: for a leaf $l$ and some $h > 0$, increase the length of the edge $e(l)$ incident to $l$ by $h$.
    \item \textbf{Deletion}: for a leaf $l$ that is not the unique child of the root, delete $l$ together with its incident edge $e(l)$. Since the parent $p$ of $l$ now has only one child remaining (and hence degree two), smooth the tree at $p$ as follows: let $e_1 = \{u, p\}$ and $e_2 = \{p, v\}$ be the edges incident to $p$, and replace them by a single edge $e = \{u, v\}$ of weight $w_e := w_{e_1} + w_{e_2}$.
    \item \textbf{Fork}: for a leaf $l$, connect two new vertices to $l$ by edges of length $0$. 
\end{itemize}

The other player, referred to as \textit{the algorithm}, responds to each operation by choosing a leaf of $T$ to occupy. If the adversary grows by $h$ the leaf $l$ where the algorithm is located, after the growth the algorithm is located at some point along the edge incident to $l$, at distance $h$ from $l$. Therefore, the algorithm needs to either move back to $l$ or choose a different leaf of $T$. If the adversary deletes the leaf where the algorithm is located, the algorithm has to move to another leaf of $T$. The algorithm pays a cost equal to the total distance it moves, and its goal is to minimize the total cost incurred during the game.

The adversary can end the game after any operation. The algorithm is said to be $\rho$-competitive if it incurs cost at most $\rho \cdot d$, where $d$ is the length of the shortest root-to-leaf path in the evolving tree at the end of the game. The instance of the game is said to have width $k$ if there are at most $k$ leaves which exist at the same time in $T$. The algorithm does not have access to $k$ in advance.

A full proof of the following reduction, which was given similarly in~\cite{randomizedLGT}, is given in Appendix \ref{sec:reduction} for completeness.

\begin{restatable}{lemma}{reductionToETG}
\label{lemma:reduction-to-ETG}
 If there exists a $\rho$-competitive algorithm for the width-$k$ evolving tree game, then there exists a $\rho$-competitive algorithm for width-$k$ layered graph traversal.
\end{restatable}

\section{Lower bound sketch}\label{sec:LB}
In this section, we give an informal sketch of the lower bound proof (Theorem \ref{th:lower-bound}). A more formal proof is deferred to Appendix~\ref{app:LB}. The construction of our lower bound will be helpful in motivating our algorithm later. We formulate the hard lower bound instance as an instance of the evolving tree game (rather than as an instance of layered graph traversal, which we will do in the more formal proof).
\paragraph{Constructing the instance.} We construct the hard instance inductively on the width $k$ of the tree. Starting from a tree which contains a single leaf, we perform a fork to obtain a tree with two branches. While the algorithm is located in one of the branches, we repeatedly start a \textit{phase} in which we play the hard instance for width $k-1$ in that branch, scaled by a sufficiently small constant. We ensure that at the end of a phase each branch contains a single leaf. When the algorithm moves to the other branch, we play phases there. We call a \emph{super-phase} a maximal consecutive sequence of phases played in the same branch. The branch where the phases are played during a super-phase is called \emph{active}, the other one \emph{passive}. Note that the passive branch is just a single edge from the root to the unique leaf of the branch. Let $\opt_{t}$ denote the length of the passive branch during super-phase $t$, and define 
\[
x := \limsup_{t \rightarrow \infty} \frac{\opt_{t+1}}{\opt_{t}}.
\]

We choose $T$ such that $\frac{\opt_{T+1}}{\opt_{T}} \approx x$, and we stop after super-phase $T$ is completed. We delete the leaf in the branch which was active in the last super-phase, so that the tree contains a single leaf in the end, connected to the root by an edge of length $\opt_{T}$.

\paragraph{Online cost analysis.}
Since each sub-instance is scaled by a small enough constant, the cost that the algorithm can save by switching branches in the middle of a phase is negligible. Thus, denoting by $A$ and $P$ the active, respectively passive, branch in the last super-phase, and by applying the inductive hypothesis, we obtain that the algorithm pays cost at least $D_{k-1} \opt_{T+1}$ during the phases in $A$, and at least $D_{k-1} \opt_{T}$ during the phases in $P$. Moreover, the algorithm pays $\opt_{t+1} + \opt_{t}$ to switch branches at the end of each super-phase $t$. By the choice of $x$, we have $\opt_{T+1} \approx x \opt_{T}$ and $\opt_{t+1} \lesssim x \opt_{t}$ for all sufficiently large $t < T$. Therefore, the overall switching cost is 
\begin{align*}
\sum_{t=1}^T \left(\opt_{t+1} + \opt_{t} \right)&= \opt_{T+1} + 2 \sum_{t=1}^T \opt_{t} \gtrsim \left(x + 2 \sum_{t=0}^{T-1} \frac{1}{x^{t}}\right) \opt_{T}. 
\end{align*}
As we can make $T$ arbitrarily large, we have 
\[
\sum_{t=0}^{T-1} \frac{1}{x^{t}} \approx \sum_{t=0}^\infty \frac{1}{x^{t}} = \frac{x}{x-1},
\]
so the overall switching cost is at least roughly 
$
\left(x +  \frac{2x}{x-1}\right) \opt_{T} = \frac{x(x+1)}{x-1} \opt_{T}
$. In total, the algorithm pays at least roughly 
\[
D_{k-1} \opt_{T+1} + D_{k-1} \opt_{T} + \frac{x(x+1)}{x-1} \opt_{T} \approx \left((x+1) D_{k-1} + \frac{x(x+1)}{x-1} \right) \opt_{T}.
\]
This expression is minimized for
$x = 1 + \sqrt{\frac{2}{1+D_{k-1}}}$, which yields the lower bound by Lemma~\ref{lem:DkViaxk} below. 

To formalize the proof, we must show that it is indeed possible to scale down sub-instances in order to make the cost incurred in an individual sub-instance negligible. To this end, we need to upper bound $\opt_{T}$. Also note that scaling a sub-instance might affect the algorithm's behavior (and thus the structure of the sub-instance, as it depends on the algorithm's behavior).

\section{Notation}
In this section, we define notation which will be useful to describe and analyze our algorithms.

We refer to an edge-weighted stemmed binary tree simply as ``tree''. If a tree $S$ contains a single leaf, we call it trivial and we represent it by an edge $e(S)$ of length $w_{e(S)}$. Otherwise, we call it non-trivial and we represent it by an edge $e(S)$ of length $w_{e(S)}$ connected to two subtrees $LS$ and $RS$ (see Fig. \ref{fig:simple-tree}). We say that $LS$ is the sibling subtree of $RS$ and vice-versa, and we say that $S$ is the parent subtree of $LS$ and $RS$. Additionally, we say that $S'$ is a subtree of $S$ if $S' = S$ or $S'$ is a subtree of $LS$ or $RS$.

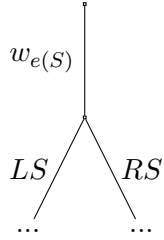
\begin{figure}[h]
\centering
\begin{tikzpicture}
\node[draw, inner sep=0pt] {} 
    child {node[draw, inner sep=0pt] {} 
        child {node {...} edge from parent node [left] {$LS$}} 
        child {node {...} edge from parent node [right] {$RS$}}
        edge from parent node [left] {$w_{e(S)}$}};
\end{tikzpicture}
\caption{Non-trivial stemmed tree S.}
\label{fig:simple-tree}
\end{figure}

Denoting by $T$ the tree on which the evolving tree game is played, we define the depth of a node as the number of edges on the path connecting that node and the root of $T$, and define the depth of $T$ as the maximum depth of one of its nodes. As in~\cite{randomizedLGT}, we parametrize the evolving tree game by depth instead of width in sections \ref{sec:almost} and \ref{sec:mainAlgo} concerned with our main algorithm. We say that an instance has depth $k$ if the depth of $T$ is at most $k$ at all times. Since the depth of a binary stemmed tree is smaller than or equal to its width, any instance of width $k$ has depth at most $k$.\footnote{On the other hand, an instance of depth $k$ can have width up to $2^{k-1}$, so our upper bounds hold even for certain instances of much larger width.} Since the parameter $k$ is not known in advance, throughout the execution of the algorithm we denote by $k$ the maximal depth of $T$ since the beginning of the game. 

We associate to each subtree $S$ of $T$ a level between $1$ and $k$: if the depth of $S$'s root is $d$, then $S$ has level $k-d$. In particular, the level of $T$ is $k$. By abuse of notation, we also associate levels to edges, so that $e(S)$ has the same level as $S$. Note that the level of a subtree $S$ can increase in two ways: 1) if $k$ is incremented (because 
the depth of $T$ reached a new maximum); 2) if the smoothing that follows a deletion causes a decrease in the depth of $S$'s root. In either case, we say that $S$ is ``promoted'' to a higher level.

For notational convenience we often treat trees like sets of points, and hence use element and set difference notation (``$\in$'' and ``$\setminus$'') accordingly. Additionally, we use $x \land y$ and $x \lor y$ as a shorthand for $\min\{x, y\}$ and $\max\{x, y\}$, respectively, and we assume that multiplication and addition take precedence over these operators.

We define $\OPT_S$ to be the shortest path from the root of $S$ to one of its leaves. So \[\OPT_S = w_{e(S)} + (\OPT_{LS} \land \OPT_{RS}).\]
We refer to our algorithm by $\ALG$, and we denote the cost incurred by the algorithm so far by $\cost(\ALG)$. To simplify notation, we assume that each tree has a distinguished leaf which describes $\ALG$'s location, so we can encode the full state of the game in a tree. 

Finally, we define
\begin{align}
x_k &:= 1 + \sqrt{\frac{2}{1 + D_{k-1}}}\qquad \text{for }k\ge 2,\label{eq:xk}
\end{align}
and note the following lemma, which we will use repeatedly throughout the rest of the paper.

\begin{lemma}\label{lem:DkViaxk}
For $D_k$ and $x_k$ as defined in~\eqref{eq:DkViaDk-1} and~\eqref{eq:xk} and $k\ge 2$, it holds that
\begin{align}
    D_k &= D_{k-1} (1 + x_k) + \frac{x_k(x_k+1)}{x_k-1}.
\end{align}
\end{lemma}
\begin{proof}
\begin{align*}
    D_{k-1} (1 + x_k) + \frac{x_k(x_k+1)}{x_k-1}&= D_{k-1}\left(2+\sqrt{\frac{2}{1+D_{k-1}}}\right)+3+\sqrt{\frac{2}{1+D_{k-1}}} + \sqrt{2\left(1+D_{k-1}\right)}\\
    &= 2D_{k-1}+\sqrt{8\left(1+D_{k-1}\right)}+3 = D_k.\qedhere
\end{align*}
\end{proof}

\section{An approach that almost works}\label{sec:almost}
In this section, we use the lower bound as inspiration to derive an algorithm and a potential function which can \textit{almost} be used to prove $D_k$-competitiveness. We parametrize the evolving tree game by depth instead of width, as the depth does not exceed the width.

Note that the lower bound is tight on the constructed instance when the algorithm applies the following strategy: Out of the two subtrees of the root, the algorithm stays in the current subtree as long as its optimum is not more than $x_k$ times larger than the optimum of the other subtree. When this condition ceases to hold, the algorithm moves to an optimal leaf in the other subtree. To move within a subtree, the algorithm applies the recursive strategy for depth $k-1$.

Let us analyze the performance of this naive algorithm. Given the current state of the tree $T$ (which includes the location of $\ALG$), we want to devise a pseudo-cost potential function $\Phi(T)$ which we can use to upper bound the cost incurred so far. If we could also show that $\Phi(T) \leq D_k \OPT_T$, this would mean that our algorithm achieves precisely the competitive ratio $D_k$ for the depth-$k$ evolving tree game.

\subsection{Potential derivation}
To derive $\Phi(T)$, we follow the intuition from the lower bound construction. The edge $e(T)$ could have been spawned by playing the game with depth $k$ and then contracting everything into a single edge; during this time, $\ALG$ could have paid up to $D_k w_{e(T)}$. 

Next, we bound the cost for switching between $LT$ and $RT$. To simplify our reasoning, suppose that the adversary is only allowed to grow the leaf where the algorithm is currently located. Further suppose that this leaf first grows by the maximal amount for which $\ALG$ stays at this leaf, and then it grows by a very small (negligible) additional amount which triggers a switch to a different leaf. Therefore, $\ALG$ moves from $LT$ to $RT$ when $\ALG$ is in the optimal leaf in $LT$ and the adversary grows this leaf so that $\OPT_{LT} = x_k \OPT_{RT}$. The algorithm pays $\OPT_{LT} + \OPT_{RT}$ for switching, of which $\OPT_{LT}$ is for backtracking to the shared root of $LT$ and $RT$ and $\OPT_{RT}$ for reaching the optimum leaf in $RT$. Suppose $\ALG$ is currently located in $RT$. Then, $\OPT_{LT}$ stayed the same since the last switch, and we can deduce that the last switch cost $\OPT_{LT}(1 + 1/x_k)$, the previous switch cost $\OPT_{LT}(1/x_k + 1/{x_k}^2)$, and so on. The total switching cost can thus be bounded by 
\[
\sum_{i=0}^{\infty} \left(\frac{1}{{x_k}^i} + \frac{1}{{x_k}^{i+1}}\right) \OPT_{LT} = \frac{x_k+1}{x_k-1} \OPT_{LT}.
\]
Of course, we swap the roles of $LT$ and $RT$ if $\ALG$ is currently in $LT$ instead of $RT$. It remains to bound the cost paid while playing the game at depth $k-1$ in the two subtrees $LT$ and $RT$. To this end, we apply an inductive argument to bound this cost by $\Phi_{k-1}(LT)$ and $\Phi_{k-1}(RT)$, respectively.

Putting everything together, for a non-trivial subtree $S$ of level $i$, we define
\begin{equation}    
\Phi_i(S) = D_i  w_{e(S)} + \frac{x_i+1}{x_i-1} \OPTOTHER(S) + \Phi_{i-1}(LS) + \Phi_{i-1}(RS), 
\label{eq:simple-potential}
\end{equation}
where 
\begin{equation}
\OPTOTHER(S) =
\left\{
	\begin{array}{ll}
	   \OPT_{LS}  & \mbox{if } \ALG \text{ is in } RS, \\
          \OPT_{RS}  & \mbox{if } \ALG \text{ is in } LS, \\
          \OPT_{LS} \lor \OPT_{RS} & \mbox{if } \ALG \text{ is not in } S.
	\end{array}
\right.
\label{eq:opt-other}
\end{equation}
Note that, if $\ALG$ is not in $S$, we defined $\OPTOTHER(S) = \OPT_{LS} \lor \OPT_{RS}$ by following the intuition that the last visited leaf in $S$ was the optimal one.

If $S$ is trivial, we simply define $\Phi_i(S) = D_i  w_{e(S)}$. We may omit the subscript in $\Phi$ if it is clear from the context.

\subsection{Analysis sketch without deletion}

\paragraph{Bounding the potential.} Assuming that 
\begin{equation}
    \OPT_{LS} \lor \OPT_{RS} \leq x_i(\OPT_{LS} \land \OPT_{RS}),
    \label{eq:ratio-invariant}
\end{equation}
a property which is indeed maintained by $\ALG$'s response to growth operations, we can prove by a simple inductive argument that $\Phi_i(S) \leq D_i \OPT_S$ for all  subtrees $S$ at level $i$. In particular, this implies $\Phi_k(T) \leq D_k \OPT_T$, as desired.

To show that $\cost(\ALG) \leq \Phi_k(T)$, it suffices to prove that the cost paid on each operation is bounded by the increase in the potential after that operation.

\paragraph{Growth.} Suppose the adversary grows the leaf $l$ where $\ALG$ is located. By unfolding the recursive formula for $\Phi(T)$, it is easy to see that it contains a term of the form $D_i w_{e(l)}$, where $w_{e(l)}$ is the length of the edge adjacent to $l$. Therefore, if $\ALG$ stays in $l$, the term $\Phi(T)$ increases by at least the cost paid by $\ALG$. Otherwise, if $\ALG$ moves to a different leaf $l'$, and if we denote by $S$ the smallest subtree containing both $l$ and $l'$, then one can show that $\OPTOTHER(S)$ increases enough to pay for the movement cost. 

\paragraph{Fork.} Since the new edges created by a fork initially have length $0$, it is easy to see that a fork does not change $\Phi(T)$ or $\cost(\ALG)$.

\subsection{Why deletion causes problems}\label{sec:deletionProblems}
 When deleting a leaf $l$ together with its incident edge, the edges in the sibling subtree of $l$ get ``promoted'' to a higher level. Although the promotion helps increase the potential, it may not be enough to counteract the edge deletion, so $\Phi(T)$ may decrease, as it happens in the example in Fig. \ref{fig:deletion-example1}. The decrease of $\Phi(T)$ is problematic, as it may break the inequality $\cost(\ALG) \leq \Phi(T)$.
\begin{figure}[h]
\centering
\begin{subfigure}[ht]{0.49\textwidth}
\begin{tikzpicture}
\node[draw, inner sep=0pt] (x) at (0,0) {} 
    child {node[draw, inner sep=0pt] at (0, 1) {} 
        child {node {$l$} 
        edge from parent node [left] {$(a+b)x_3$}}
        child {node[draw, inner sep=0pt] (y) {} 
            child {node {$\ALG$} edge from parent node [left] {$a$}}
            child {node {} edge from parent node [right] {$a$}}
        edge from parent node [right] {$b$}}
    edge from parent node [left] {$0$}};
\node[draw, inner sep=0pt] at (3,0) {} 
    child {node[draw, inner sep=0pt] at (0,-0.6) {} 
        child {node {$\ALG$} 
        edge from parent node [left] {$a$}}
        child {node {} 
        edge from parent node [right] {$a$}}
    edge from parent node [left] {$b$}};

\draw[->, thick] (1, -1) -- (2, -1) node[midway, above, scale=0.8] {delete};
\end{tikzpicture}
\caption{Delete $l$ in $T_1$ to obtain $T'_1$.}
\label{fig:deletion-example1}
\end{subfigure}
\begin{subfigure}[ht]{0.49\textwidth}
\begin{tikzpicture}
\node[draw, inner sep=0pt] (x) at (0,0) {} 
    child {node[draw, inner sep=0pt] at (0, 1) {} 
        child {node {$l$} 
        edge from parent node [left] {$(a+b)x_3$}}
        child {node[draw, inner sep=0pt] (y) {} 
            child {node {$\ALG$} edge from parent node [left] {$a$}}
            child {node {} edge from parent node [right] {$ax_2$}}
        edge from parent node [right] {$b$}}
    edge from parent node [left] {$0$}};
\node[draw, inner sep=0pt] at (3,0) {} 
    child {node[draw, inner sep=0pt] at (0,-0.6) {} 
        child {node {$\ALG$} 
        edge from parent node [left] {$a$}}
        child {node {} 
        edge from parent node [right] {$ax_2$}}
    edge from parent node [left] {$b$}};
\draw[->, thick] (1, -1) -- (2, -1) node[midway, above, scale=0.8] {delete};
\end{tikzpicture}
\caption{Delete $l$ in $T_2$ to obtain $T'_2$.}
\label{fig:deletion-example2}
\end{subfigure}
\caption{Deletion examples. $\ALG$ represents the position of the algorithm and the labels next to the edges indicate weights.}
\label{fig:deletion-examples}
\end{figure}
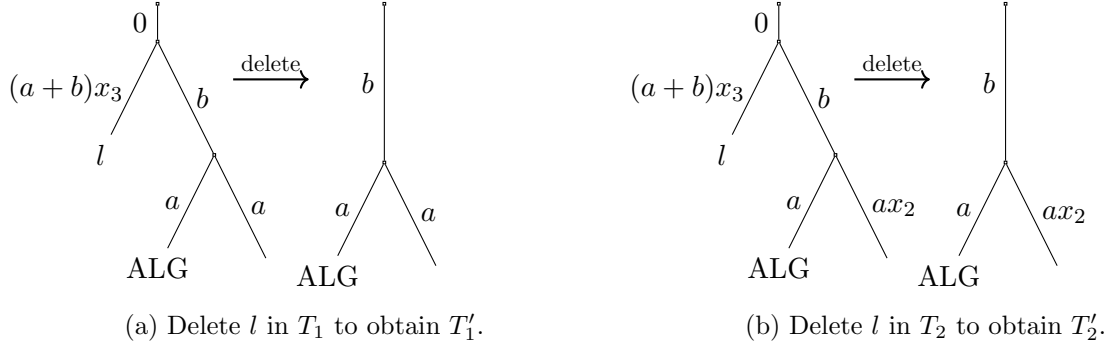
\begin{align*}
\Phi(T_1) &= D_2 (a+b) x_3 + D_2 b + 2 D_1 a + \frac{x_3+1}{x_3-1} (a+b) x_3 + \frac{x_2+1}{x_2-1} a  \\
\Phi(T'_1) &= D_3 b + 2 D_2 a + \frac{x_3+1}{x_3-1} a \\
\Phi(T'_1) - \Phi(T_1) &= \left(2D_2 + \frac{x_3+1}{x_3-1} - D_2 x_3 - 2 D_1 - \frac{x_3(x_3+1)}{x_3-1} - \frac{x_2+1}{x_2-1}\right)a \approx  -2.4 a < 0.
\end{align*} 

On the other hand, the potential may increase too much after deletion, violating the inequality $\Phi(T) \leq D_k \OPT_T$. Recall that this inequality is satisfied if \eqref{eq:ratio-invariant} holds for all subtrees $S$ at level $i$.  However, if some $S$ is promoted from level $i$ to $i+1$, \eqref{eq:ratio-invariant} may no longer be satisfied, as $x_{i+1} < x_i$. This happens in the example in Fig. \ref{fig:deletion-example2}, where $\OPT_{RT'_2} > x_3 \OPT_{LT'_2}$ (because $ax_2 > ax_3$).

So what went wrong with our derivation of the potential? Note that we implicitly assumed that deletion always occurs in a leaf whose sibling is also a leaf, so only trivial subtrees get promoted. Indeed, the lower bound instance is restricted to this kind of ``easy'' deletions, and it can be proved that the potential is well-behaved on ``easy'' deletions. By contrast, in both examples in Fig. \ref{fig:deletion-examples} we delete a leaf whose sibling is a non-trivial subtree. Indeed, the core difficulty of the problem lies in handling deletions in the general case.

\section{Refining the algorithm}\label{sec:mainAlgo}

We now provide an algorithm and analysis that achieves competitive ratio $O(2^k)$, proving Theorem~\ref{th:upper-bound}. It is based on two refinements, \emph{forgetting} and \emph{imbalancing}, that make the analysis better aligned with the structure of hard instances. Conceptually, forgetting corresponds to \emph{temporal} canonicalization by limiting the effect of information that is not currently relevant, while imbalancing achieves \emph{geometric} canonicalization by reshaping the instance towards a worst-case form.

\paragraph{Forgetting.} \emph{Forgetting} stabilizes the potential by truncating certain terms, effectively redefining it as if some parts of the tree had not yet been revealed.
In the example in Fig.~\ref{fig:deletion-example2}, we could pretend that the edge in $RT'_2$ had length $a x_3$ instead of $a x_2$, obtaining the tree $T''$ from Fig.~\ref{fig:extremely-imbalanced-tree}, which satisfies $\Phi(T'') = D_3 \OPT_{T''}$. Intuitively, this corresponds to \emph{rewinding} part of the exploration.
One could extend the evolving tree game by introducing a new operation in which the algorithm could choose a subtree $S$ and some $d>0$, and forget everything in $S$ at distance more than $d$ from the root of $S$. However, this would make the algorithm too powerful, as repeated forgetting could prevent the adversary from growing the tree.
Instead, we simulate the same effect analytically by capping certain terms in the potential function, ensuring bounded growth after deletions.

\paragraph{Imbalancing.} \emph{Imbalancing} enforces a structural regularity that serves to better align an instance with typical worst-case scenarios.
It is based on the observation that imbalanced trees act as canonical hard instances, since for such trees the potential achieves its maximum value $\Phi_i(S) = D_i \OPT_S$.
To move toward this canonical form, we slightly distort some edge lengths to create a controlled imbalance.
In the example in Fig.~\ref{fig:deletion-example1}, distorting the rightmost edge in $T'_1$ by a factor of $x_3$ yields the tree $T''$ from Fig.~\ref{fig:extremely-imbalanced-tree}, satisfying $\Phi(T'') = D_3 \OPT_{T''} = D_3 \OPT_{T_1}$.
More generally, one can show that if a tree is very imbalanced, then it is an extreme tree (i.e., a tree $S$ for which the pseudo-cost $\Phi_i(S)$ attains the upper bound of $D_i \OPT_S$, reflecting that it is a hard instance).

To achieve a high imbalance, we can perform the following operation: \textit{scale} the length of an edge $e$ by a factor $f \geq 1$. If $T^0$ is the tree on which the game is played, we maintain a distorted tree $T$ which has the same vertices and edges as $T^0$, but different edge weights. When the adversary performs some operation on $T^0$, we perform the same operation on $T$. Additionally, we may perform imbalancing operations on $T$. Provided the overall distortion is not too large, this only loses a constant factor.

\begin{figure}[ht]
\centering
\begin{tikzpicture}
\node[draw, inner sep=0pt] {} 
    child {node[draw, inner sep=0pt] {} 
        child {node {$\ALG$} 
        edge from parent node [left] {$a$}}
        child {node {} 
        edge from parent node [right] {$ax_3$}}
    edge from parent node [left] {$b$}};
\end{tikzpicture}
\caption{Extreme tree $T''$. $\ALG$ represents the position of the algorithm and the labels next to the edges indicate weights.}
\label{fig:extremely-imbalanced-tree}
\end{figure}
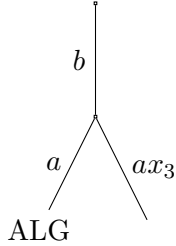

\subsection{Algorithm description}
\label{sec:algo-description}
We are now ready to describe our algorithm. Recall that the tree $T^0$ on which the game is played is initialized to a have a single leaf, connected to the root by a zero-length edge. We initialize the distorted tree $T$ in the same manner. Since the depth $k$ of the instance is not known in advance, we initialize $k:=1$, and we increment $k$ when the depth of $T$ exceeds the current value of $k$. 

We denote by $w^0$ and $w$ the edge weights in $T^0$ and $T$, respectively. We denote by $S^0$ a subtree of $T^0$, and by $S$ the subtree corresponding to $S^0$ in $T$.

We will ensure that the following invariant always holds.

\paragraph{Ratio invariant.} Let $l_{\ALG}$ be the leaf where $\ALG$ is located, and let $S$ be a non-trivial subtree of $T$ of level $i$ such that $l_{\ALG} \in S$. Then, if $l_\ALG \in LS$ we have $\OPT_{LS} \leq x_i \OPT_{RS}$, and if $l_\ALG \in RS$ we have $\OPT_{RS} \leq x_i \OPT_{LS}$.

Based on the operation performed by the adversary on $T^0$, we describe our algorithm's response on $\TT$.

\paragraph{Growth.} When the adversary grows a leaf and the ratio invariant remains satisfied, $\ALG$ stays in its current position. Otherwise, let $S$ be the highest-level subtree for which the ratio invariant inequality is violated. Then $\ALG$ moves to the optimal leaf in $S$. Observe that the ratio invariant is now satisfied.

\paragraph{Deletion.}
When the adversary deletes a leaf $l$, we proceed as follows:
\begin{enumerate}
\item Pretend that the length of $e(l)$ grows to $\infty$, moving accordingly. 
\item Let $S$ be the parent subtree of $l$. If $\ALG$ is in $S$, $\ALG$ moves to an optimal leaf in $S$.
\item Delete $l$ and its incident edge.  
\item Transform $S$ into an extreme tree (see Section \ref{sec:extreme-imbalance}). 
\end{enumerate}

\paragraph{Fork.} If $\ALG$ is located in the leaf $l$ which is forked, $\ALG$ moves to any of the new leaves. Additionally, if the depth of the tree becomes larger than $k$, then we increment $k$, $\ALG$ moves to the optimal leaf, and we transform $\TT$ into an extreme tree.

\subsection{Potential}
\label{sec:updated-potential}

We now show how to adapt the potential to achieve the effect of ``forgetting'' parts of the tree.

\paragraph{Refined potential.} For a non-trivial subtree $S$ of $\TT$ of level $i$, define
\begin{equation}
\Phi_i(S) = D_i w_{e(S)} + \frac{x_i+1}{x_i-1} \overline{\OPTOTHER(S)} + \overline{\Phi_{i-1}(LS)} + \overline{\Phi_{i-1}(RS)},
\label{eq:updated-potential}
\end{equation}
where \begin{align}
\overline{\OPTOTHER(S)} &= \OPTOTHER(S) \land x_i (\OPT_{LS} \land \OPT_{RS}) \label{eq:opt-other-cap}\\
\overline{\Phi_{i-1}(XS)} &= \Phi_{i-1}(XS) \land D_{i-1}x_i \left(\OPT_{LS} \land \OPT_{RS}\right) \label{eq:phi-cap} \quad \text{for } X \in \{L, R\}.
\end{align}
 For a trivial subtree $S$, define $\Phi_i(S) = D_i w_{e(S)}$.

 Note that, as long as $\OPT_{LS} \lor \OPT_{RS} \leq x_i(\OPT_{LS} \land \OPT_{RS})$, by Claim \ref{claim:bounding-the-potential} below we have \[\overline{\OPTOTHER(S)} = \OPTOTHER(S), \quad \overline{\Phi_{i-1}(RS)} = \Phi_{i-1}(RS),  \quad \text{and} \quad \overline{\Phi_{i-1}(LS)} = \Phi_{i-1}(LS),\]
 so we recover the same recursive formula for $\Phi_i(S)$ as in the previous version of the potential in \eqref{eq:simple-potential}. Otherwise, if, say, $\OPT_{LS} > x_i\OPT_{RS}$, the refined potential caps some terms -- as if ``forgetting'' some part of $LS$ so that $\OPT_{LS}$ gets reduced to $x_i\OPT_{RS}$.

It is also important to note that the subtrees which contain $\ALG$'s location will not have their potential capped, as shown in the following claim.
 
\begin{claim}
    For all subtrees $S$ at level $i<k$ such that $\ALG$ is in $S$, we have $\overline{\Phi_i(S)} = \Phi_i(S)$.
    \label{claim:not-capped}
\end{claim}
\begin{proof}
    Let $P$ be the parent subtree of $S$ and suppose without loss of generality that $S = LP$. By Claim \ref{claim:bounding-the-potential} below, we have $\Phi_i(LP) \leq D_i \OPT_{LP}$, and by the ratio invariant we have $\OPT_{LP} \leq x_{i+1}\OPT_{RP}$. Thus, $\Phi_i(LP) \leq x_{i+1} D_i (\OPT_{LP} \land \OPT_{RP})$, so $\overline{\Phi_i(LP)} = \Phi_i(LP)$.
\end{proof}

\subsection{Analysis}
We will ensure that each edge in the distorted tree $\TT$ is scaled by a factor between $1$ and some constant $C$. Thus, $\OPT_{\TT} \leq C \cdot \OPT_{T^0}$. Moreover, when $\ALG$ moves from a leaf to another, we will charge the movement cost according to the distances in $\TT$ instead of the distances in $T^0$, which means that we are overestimating the cost paid. 

In the rest of this section we will prove that 
\[
\cost(\ALG) \leq \Phi(\TT) \leq D_k \OPT_{\TT}, 
\]
which together with the previous observations and the fact that $D_k = O(2^k)$ (see Appendix \ref{sec:bounding-D}) yields
\[\cost(\ALG) \leq C \cdot D_k \cdot \OPT_{T^0} = O(2^k) \cdot\OPT_{T^0}.\]

We denote by $\Delta \cost(\ALG)$ the cost paid in the current operation and by $\Delta \Phi(\TT)$ the increase of the potential in the current operation. To ensure that $\cost(\ALG) \leq \Phi(\TT)$, it suffices to show that $\Delta \cost(\ALG) \leq \Delta \Phi(\TT)$ holds for all operations. 

For two leaves $l, l' \in S$, we denote by $\dist_S(l, l')$ the distance between $l$ and $l'$ in subtree $S$. We denote by $l_{\OPT_S}$ the optimal leaf in $S$ (breaking ties arbitrarily).

\subsubsection{Bounding the potential}
\begin{claim}
For any subtree $S$ of level $i$, it holds that $\Phi_i(S) \leq D_i \OPT_S$.\label{claim:bounding-the-potential}
\end{claim}
\begin{proof}
    By induction. For a trivial subtree $S$, we have $\Phi_i(S) = D_i w_{e(S)} = D_i \OPT_S$. For the inductive step, suppose without loss of generality that $\OPT_{LS} \leq \OPT_{RS}$, so $\OPT_S =  w_{e(S)} + \OPT_{LS}$. By the inductive hypothesis, $\Phi_{i-1}(LS) \leq D_{i-1} \OPT_{LS}$. Using this together with \eqref{eq:opt-other-cap} and \eqref{eq:phi-cap}, we get
    \begin{align*}
        \Phi_i(S) &= D_i w_{e(S)} + \frac{x_i+1}{x_i-1} \overline{\OPTOTHER(S)} + \overline{\Phi_{i-1}(LS)} + \overline{\Phi_{i-1}(RS)} \\ 
        &\leq D_i w_{e(S)} + \frac{x_i+1}{x_i-1} x_i \OPT_{LS} + D_{i-1} \OPT_{LS} + D_{i-1} x_i \OPT_{LS} \\ 
        &= D_i (w_{e(S)} + \OPT_{LS}) = D_i \OPT_S,
    \end{align*}
    where in the penultimate step we plugged in the formula for $D_i$.
\end{proof}

In the next sections, we show that $\Delta \cost(\ALG) \leq \Delta \Phi(\TT)$ holds for all possible operations and that the ratio invariant is maintained.

\subsubsection{Growth}
\label{sec:growth-analysis}
By the definition of the algorithm, one can see that the ratio invariant is maintained. It remains to show that the potential does not decrease after a growth operation. We first prove this in the following case.

\begin{claim}
Suppose that the adversary grows a leaf $l_g$ by $h$, and $\ALG$ is located at $l_g$ before the growth operation. Then, $\Delta \cost(\ALG) \leq \Delta \Phi(\TT)$.
\label{claim:growth-works-easy-case}
\end{claim}

\begin{proof}
    If the algorithm remains at $l_g$ after the growth, since $w_{e(l_g)}$ increases by $h$, we have $\Delta \Phi(S_{l_{g}}) = D_i h$, where $S_{l_{g}}$ is the trivial subtree containing $l_g$ and $i$ is its level. By repeatedly applying Claim \ref{claim:not-capped}, we obtain \[\Delta \Phi({T}) \geq D_i h \geq h = \Delta \cost(\ALG).\]

    It remains to consider the case when $\ALG$ moves from its initial position $l_{g}$ to another leaf. Let $S$ be the highest-level subtree before the growth for which the ratio invariant inequality in the corresponding subtree $S'$ after growth is violated, and let $i$ be the level of $S$ and $S'$. Suppose without loss of generality that $\ALG$ is currently at $l_{g} \in LS$ and moves to an optimal leaf in $RS'=RS$. Note that this movement can only be triggered when $l_g$ is the (unique) optimum leaf in $LS$, and therefore the movement cost is $\OPT_{LS} + \OPT_{RS}$.

    Since the ratio invariant is satisfied before the growth but violated after the growth, we have 
\begin{align}
    \OPT_{LS} &\leq x_i \OPT_{RS} \label{eq:g1a}
    \\
    \OPT_{LS'} &> x_i \OPT_{RS}. \label{eq:g2a}
\end{align}

Therefore, before the growth we have \[\overline{\OPTOTHER(S)} \leq \OPTOTHER(S) = \OPT_{RS},\] and after the growth, by \eqref{eq:g2a}, we have \[\overline{\OPTOTHER(S')} = \OPT_{LS'} \land (x_i \OPT_{RS}) = x_i \OPT_{RS}.\] Therefore, $\overline{\OPTOTHER(S)}$ increases by at least $(x_i-1) \OPT_{RS}$. Notice that no term in $\Phi(\TT)$ decreases after the growth, so by repeatedly applying Claim \ref{claim:not-capped}, we have 
\begin{align*}
\Delta \Phi(\TT) 
\geq \Delta \Phi(S)  &\geq \frac{x_i+1}{x_i-1}(x_i-1) \OPT_{RS} \\
&= (x_i+1) \OPT_{RS}
\\
& \stackrel{\eqref{eq:g1a}}{\geq} \OPT_{RS} + \OPT_{LS} = \Delta \cost(\ALG).\qedhere
\end{align*}
\end{proof}

The following claim provides a stronger bound on the potential, and will be used repeatedly in our analysis.
\begin{claim}\label{claim:bounding-the-potential2}
Let $l_{\ALG}$ be the leaf where $\ALG$ is located, let $S$ be a subtree of level $i$ such that $l_\ALG \in S$, and let $l_{\OPT_S}$ be an optimal leaf in $S$. Then, 
\[\Phi_i(S) + \dist_S(l_\ALG, l_{\OPT_S}) \leq D_i \OPT_S.\]
\end{claim}
\begin{proof}
    Imagine repeatedly growing the leaf where $\ALG$ is located until eventually $\ALG$ moves to $l_{\OPT_S}$. By Claim \ref{claim:growth-works-easy-case}, by the time $\ALG$ reaches $l_{\OPT_S}$, $\Phi(T)$ would have increased by at least the movement cost paid by $\ALG$, which is at least $\dist_S(l_\ALG, l_{\OPT_S})$ by the triangle inequality. Moreover, the potential at the end of this process would still be bounded by $D_i \OPT_S$ by Claim \ref{claim:bounding-the-potential}, which concludes the proof.
\end{proof}

We next prove the following auxiliary claim, which will be used to show that the potential does not decrease after any growth operation.
\begin{claim}
Suppose that in response to a growth operation $\ALG$ moves from its initial position $l_{\ALG} \in LS$ to an optimal leaf in $RS$. Let $P$ be a subtree of $LS$ such that $l_{\ALG} \in P$. Then, the change in potential after the growth and $\ALG$'s movement satisfies
    \[
    \Delta \overline{\Phi(P)} \geq A_P - \OPT_P,
    \]
    where $A_P = \dist_P(\rooot_P, l_{\ALG})$ and $\rooot_P$ denotes the root of $P$.
    \label{claim:phi-bar-increases-enough}
\end{claim}
\begin{proof}
    By induction. If $P$ is a trivial subtree, then $A_P - \OPT_P = 0 \leq \Delta \overline{\Phi(P)}$ and we are done. Thus, suppose $P$ is a non-trivial subtree of level $i$, and let $P'$ denote the subtree $P$ after the growth. We begin by showing that it suffices to prove that $\Delta \Phi(P)\geq A_P - \OPT_P$.  By Claim~\ref{claim:bounding-the-potential2}, we have 
    \begin{equation}
    \overline{\Phi(P)} \leq \Phi(P) \leq D_i \OPT_P - \dist(l_{\ALG}, l_{\OPT_P}).
    \label{eq:1000}
    \end{equation}

    By the triangle inequality, 
    \begin{align*}
     \dist(l_{\ALG}, l_{\OPT_P}) &\ge \dist(\rooot_P, l_{\ALG}) - \dist(\rooot_P, l_{\OPT_P})\\
     &= A_P - \OPT_P.
    \end{align*}

    Plugging this into \eqref{eq:1000}, we get 
    \begin{equation}
    \overline{\Phi(P)} \leq \Phi(P) \leq D_i \OPT_P - A_P + \OPT_P.
    \label{eq:1001}
    \end{equation}

    Let $Q$ be the sibling of $P$, and $Q'$ be the subtree $Q$ after the growth. 
    Since the ratio invariant was satisfied before the growth, we have
    \[x_{i+1} \OPT_{Q'} \geq x_{i+1} \OPT_Q \geq \OPT_P, \]
    and substituting into \eqref{eq:phi-cap} yields \[
    \overline{\Phi(P')} \geq \Phi(P') \land D_i\OPT_P.
    \]
    Combining this with \eqref{eq:1001}, we get 
    \[
    \Delta \overline{\Phi(P)} \geq (A_P - \OPT_P) \land \Delta \Phi(P).
    \]
    Therefore, it suffices to show that $\Delta \Phi(P)\geq A_P - \OPT_P$. 
    Suppose without loss of generality that $l_{\ALG} \in LP$. By the inductive hypothesis, 
    \begin{equation}
        \Delta \overline{\Phi(LP)} \geq A_{LP} - \OPT_{LP}.
    \label{eq:1002}
    \end{equation}
    If $\OPT_P$ is in $LP$, then $A_P - \OPT_P = A_{LP} - \OPT_{LP}$, and we can conclude by substituting this into \eqref{eq:1002}.
    Therefore, suppose that $l_{\OPT_P}$ is in $RP$. Then, 
    \begin{equation}
    A_P - \OPT_P = (A_{LP} - \OPT_{LP}) + (\OPT_{LP} - \OPT_{RP}).
    \label{eq:1003}
    \end{equation}

    The first term is bounded by \eqref{eq:1002}, and we will show that the second term is bounded by $\Delta  \overline{\OPTOTHER(P)}$. To this end, note that
    \[
    \overline{\OPTOTHER(P)} \leq \OPTOTHER(P) = \OPT_{RP}
    \]
    and 
    \[
    \overline{\OPTOTHER(P')} = (\OPT_{LP'} \lor \OPT_{RP'}) \land (x_i (\OPT_{LP'} \land \OPT_{RP'})) \geq \OPT_{LP},
    \]
    since $\OPT_{LP'} \geq \OPT_{LP}$ and \[x_i \OPT_{RP'} \geq x_i \OPT_{RP} \geq \OPT_{LP}\] by the ratio invariant. Thus, 
    \begin{equation}
    \Delta  \overline{\OPTOTHER(P)} \geq \OPT_{LP} - \OPT_{RP}. 
    \label{eq:1004}
    \end{equation}

    Combining \eqref{eq:1004}, \eqref{eq:1003}, \eqref{eq:1002}, and \eqref{eq:updated-potential}, and noting that $\overline{\OPTOTHER(P)}$ has coefficient $\frac{x_i+1}{x_i-1} > 1$, we get
    \[
    \Delta \Phi(P) \geq A_P - \OPT_P,
    \]
    as desired.    
\end{proof}

We are now ready to prove that potential does not decrease after any growth operation. 
\begin{claim}
For any growth operation, we have $\Delta \cost(\ALG) \leq \Delta \Phi(\TT)$.
\label{claim:growth-works}
\end{claim}

\begin{proof}
Suppose that the adversary grows a leaf $l_g$ by $h$, and $\ALG$ is located at $l_\ALG$ before the growth operation. The case when $l_\ALG = l_g$ is already covered Claim \ref{claim:growth-works-easy-case}, so it remains to consider the case when $l_\ALG \neq l_g$.
    
If the algorithm does not move and stays in the same leaf $l_{\ALG}$ before and after the growth, we have $\Delta \Phi(T) \geq 0 = \Delta \cost(\ALG)$ and we are done. Therefore, suppose $\ALG$ moves from its initial position $l_{\ALG}$ to another leaf. Let $S$ be the highest-level subtree before the growth for which the ratio invariant inequality in the corresponding subtree $S'$ after growth is violated, and let $i$ be the level of $S$ and $S'$. Suppose without loss of generality that $\ALG$ is currently in $l_{\ALG} \in LS$ and moves to an optimal leaf in $RS'=RS$. Note that this movement can only be triggered when the growing leaf $l_g$ is the (unique) optimum leaf in $LS$. Let $A_{LS}$ be the distance from $l_{\ALG}$ to the root of subtree $LS$, so that $\ALG$ pays $A_{LS} + \OPT_{RS}$ for the movement. By the auxiliary Claim \ref{claim:phi-bar-increases-enough},
    \begin{equation}   
    \Delta \overline{\Phi(LS)} \geq A_{LS} - \OPT_{LS}.
    \label{eq:phi-bar-increases}
    \end{equation}

Similarly to the proof of Claim \ref{claim:growth-works-easy-case}, we have
that $\overline{\OPTOTHER(S)}$ increases by at least $(x_i-1) \OPT_{RS}$. Notice that no term in $\Phi(\TT)$ decreases after the growth, so by repeatedly applying Claim \ref{claim:not-capped}, and using \eqref{eq:phi-bar-increases}, we have 
\begin{align*}
\Delta \Phi(\TT) 
\geq \Delta \Phi(S)  &\geq \frac{x_i+1}{x_i-1}(x_i-1) \OPT_{RS} +  A_{LS} - \OPT_{LS} \\
&= (x_i+1) \OPT_{RS} + A_{LS} - \OPT_{LS}
\end{align*}

Since the ratio invariant was satisfied before the growth, we have $\OPT_{LS} \leq x_i \OPT_{RS}$, and thus 
\[
\Delta \Phi(\TT) \geq \OPT_{RS} + A_{LS} = \Delta \cost(\ALG). \qedhere
\]
\end{proof}

\subsubsection{Deletion}
\label{sec:deletion-analysis}
This section is concerned with the analysis for the deletion operation. Let $l$ be the leaf deleted by the adversary. Without loss of generality, assume that $w_{e(l)} > 2 \OPT_{T}$. If this is not the case, we can pretend that $l$ grows by a sufficient amount just before the deletion is announced. Our assumption ensures that $\ALG$ cannot be located in $l$, by the ratio invariant. 

Let $S$ be the parent subtree of $l$, as shown in Fig. \ref{fig:deletion-stages}. Recall that we take the following steps in order to handle the deletion:
\begin{enumerate}
\item If $\ALG$ is in $S$, $\ALG$ moves to an optimal leaf in $S$. We denote the resulting subtree by $S'$.
\item Delete $l$ together with its incident edge, and increase the levels of all remaining subtrees in $S'$. Call the resulting subtree $S''$.  
\item Apply the extreme-imbalance procedure (described below) on $S''$ to obtain extreme subtree $S'''$. 
\end{enumerate}

\begin{figure}[ht]
\centering
\begin{subfigure}[t]{0.24\textwidth}
\begin{tikzpicture}
\node[draw, inner sep=0pt] (step1) at (0, 0) {}
    child {node[draw, inner sep=0pt] {} 
        child {node {$l$} edge from parent node [left] {$b$}} 
        child {node[draw, inner sep=0pt] {} 
        child {node {\dots}
        child {node {$\OPT$} edge from parent node [left] {$f$}}
        edge from parent node [left] {{}}
               edge from parent node [left] {$d$}}
        child {node {\dots}
        child {node {$\ALG$} edge from parent node [right] {$g$}}
        edge from parent node [left] {{}}
               edge from parent node [right] {$e$}}
        edge from parent node [right] {$c$}}
        edge from parent node [left] {$a$}};
\end{tikzpicture}
\caption{Subtree $S$}
\end{subfigure}
\begin{subfigure}[t]{0.24\textwidth}
\begin{tikzpicture}
\node[draw, inner sep=0pt] {} 
    child {node[draw, inner sep=0pt] {} 
        child {node {$l$} edge from parent node [left] {$b$}} 
        child {node[draw, inner sep=0pt] {}
        child {node {\dots}
        child {node[text width=1cm,align=center]  { ALG \\ OPT} edge from parent node [left] {$f$}}
        edge from parent node [left] {{}}
               edge from parent node [left] {$d$}}
        child {node {\dots}
        child {node[text width=1cm,align=center]  { \break \break } edge from parent node [right] {$g$}}
        edge from parent node [left] {{}}
               edge from parent node [right] {$e$}}
    edge from parent node [right] {$c$}}
    edge from parent node [left] {$a$}};
\end{tikzpicture}
\caption{Subtree $S'$}
\end{subfigure}
\begin{subfigure}[t]{0.24\textwidth}
    \begin{tikzpicture}
\node[draw, inner sep=0pt] (A) {} 
  child {node[draw, inner sep=0pt] (B) at (0, -1.5) {} 
        child {node {\dots}
        child {node[text width=1cm,align=center]  { ALG \\ OPT} edge from parent node [left] {$f$}}
        edge from parent node [left] {{}}
               edge from parent node [left] {$d$}}
        child {node {\dots}
        child {node[text width=1cm,align=center]  {\break \break} edge from parent node [right] {$g$}}
        edge from parent node [left] {{}} edge from parent node [right] {$e$}}
        edge from parent node [left] {$a+c$}};
\end{tikzpicture}
\caption{Subtree $S''$}
\end{subfigure}
\begin{subfigure}[t]{0.24\textwidth}
    \begin{tikzpicture}
\node[draw, inner sep=0pt] (A) {} 
    child {node[draw, inner sep=0pt] (B) at (0, -1.5) {} 
        child {node {\dots}
        child {node[text width=1cm,align=center]  { ALG \\ OPT} edge from parent node [left] {$f$}}
        edge from parent node [left] {{}}
               edge from parent node [left] {$d$}} 
        child {node {\dots}
        child {node[text width=1cm,align=center]  {\break \break} edge from parent node [right] {$g^{new}$}}
        edge from parent node [left] {{}} edge from parent node [right] {$e^{new}$}}
        edge from parent node [left] {$a+c$}};
\end{tikzpicture}
\caption{Subtree $S'''$}
\end{subfigure}
\caption{Stages of deletion. $\OPT$ denotes the location of the optimal leaf in the subtree and $\ALG$ denotes the location of the algorithm (if the algorithm is in the subtree). The edges are labelled by their weights.}
\label{fig:deletion-stages}
\end{figure}
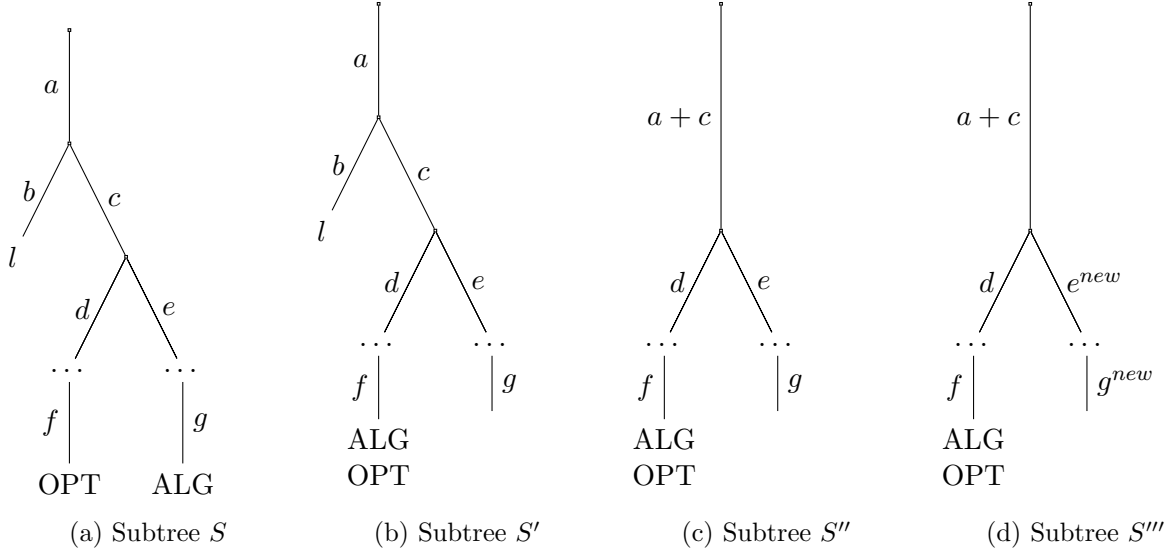

Recall that a subtree $N$ of level $j$ is said to be extreme if the inequality in Claim \ref{claim:bounding-the-potential} is tight, i.e.\ $\Phi_j(N) = D_j \OPT_N$. The following recursive procedure can be used to make a subtree extreme.

\paragraph{Extreme-imbalance procedure on a subtree $\boldsymbol{N}$.}
\label{sec:extreme-imbalance}
If $N$ is a trivial subtree, leave $N$ unchanged. Otherwise, if $N$ is a non-trivial subtree of level $j$, and assuming without loss of generality that the optimal leaf of $N$ is in $LN$, do:
\begin{enumerate}
    \item Apply the extreme-imbalance procedure to $LN$ and $RN$.
    \item If $\OPT_{RN} \leq x_j \OPT_{LN}$, scale up all edges in $RN$ by $x_j\frac{ \OPT_{LN}}{\OPT_{RN}} \in [1, x_j]$.
\end{enumerate}

\begin{claim}
Assuming that $\ALG$ is located either in $\TT \setminus N$ or in the optimal leaf of $N$, the extreme-imbalance procedure transforms $N$ into a subtree $N'$ such that
\begin{itemize}
    \item $\Phi_j(N') = D_j \OPT_{N'}$ (i.e., $N'$ is extreme), and
    \item $\OPT_{N'} = \OPT_N$ (and more precisely, no edge on the optimal path is altered).
\end{itemize}
\label{claim:extreme-procedure}
\end{claim}
\begin{proof}
If $N$ is trivial, then $N$ is already extreme, so we are done. Therefore, suppose $N$ is non-trivial. Since $\ALG$ is located either in the optimal leaf of $LN$ or not in $N$ at all, we can apply the procedure recursively to make $LN$ and $RN$ extreme. In the end, we have $\OPT_{LN'} = \OPT_{LN}$ and $\OPT_{RN'} \geq x_j \OPT_{LN}$. Putting everything together and using \eqref{eq:updated-potential}, a simple calculation shows that $\Phi_j(N') = D_j \OPT_N$, as desired.
\end{proof}

Let $\TT^{new}$ be the tree obtained from $\TT$ by replacing $S$ with $S'''$. Let $i$ be the level of $S$, $S'$, $S''$, and $S'''$.

\begin{claim}
We have
\[
\Delta \Phi(\TT) = \Phi(\TT^{new}) - \Phi(\TT) \geq \Delta \cost(\ALG).
\]
\label{claim:potential-increases-for-deletion}
\end{claim}
\begin{proof}
Suppose that initially $\ALG$ is in a leaf $l_\ALG \notin S$. Then, $\Delta \cost(\ALG) = 0$, since $\ALG$ does not move during any step. We have $\OPT_S = \OPT_{S'} = \OPT_{S''}$, as the deleted leaf $l$ was not optimal in $S$ (by our assumption that $w_{e(l)} > 2 \OPT_{T}$). By Claim \ref{claim:extreme-procedure}, we obtain \[\Phi(S''') = D_i \OPT_{S'''} = D_i \OPT_S,\] and by Claim \ref{claim:bounding-the-potential} we have \[\Phi(S) \leq D_i \OPT_S.\] Thus, $\Phi(S''') - \Phi(S) \geq 0$, and so $\Delta \Phi(\TT) \geq 0$.

Now suppose that initially $\ALG$ is in a leaf $l_\ALG \in S$. Then, $\ALG$ pays $\Delta \cost(\ALG) = \dist_S(l_\ALG, l_{\OPT_S})$ to move from $l$ to the optimal leaf $l_{\OPT_S}$. By Claim \ref{claim:bounding-the-potential2} we have \[\Phi(S) + \dist_S(l_\ALG, l_{\OPT_S}) \leq D_i \OPT_S,\]
and by Claim \ref{claim:extreme-procedure} we have 
\[
\Phi(S''') = D_i \OPT_{S'''} = D_i \OPT_S.
\]
Thus, $\Phi(S''') - \Phi(S) \geq \Delta \cost(\ALG)$. Since $\ALG$ is initially in $S$ and it remains in $S'''$, by repeatedly applying Claim \ref{claim:not-capped}, we obtain $\Delta \Phi(\TT) = \Phi(S''') - \Phi(S)$, so $\Delta \Phi(\TT) \geq \Delta \cost(\ALG)$.
\end{proof}

\begin{claim}
    $\TT^{new}$ satisfies the ratio invariant.
    \label{claim:ratio-invariant-satisfied}
\end{claim}
\begin{proof}
    We know that the ratio invariant is satisfied in the beginning for $\TT$. $\ALG$ may move to an optimal leaf in $S$ at step 1, but this maintains the invariant. The levels of the subtrees of $S'$ increase at step 2, so the $x_j$'s used for the ratio invariant inequality inside $S'$ and $S''$ are different. But the ratio invariant inequality is satisfied trivially inside $S''$ because if $\ALG$ is in $S''$, then it is in an optimal leaf in $S''$. At step 3, the lengths of the edges on the path from the root to $\ALG$'s location do not change, whereas the lengths of other edges could only increase; thus, the invariant is preserved.
\end{proof}

\subsubsection{Fork}
\label{sec:forking-analysis}
If the fork increases the depth of $\TT$ from $k$ to $k+1$, we increment $k$. This increases the levels of all subtrees, which could make the ratio invariant become violated. To fix this, $\ALG$ moves from its current leaf $l_\ALG$ to an optimal leaf in $\TT$, and we apply the extreme-imbalance procedure on $\TT$ to obtain a new tree $\TT^{new}$. The movement cost paid by $\ALG$ is $\Delta \cost(\ALG) = \dist_{\TT}(l_\ALG, l_{\OPT_{\TT}})$, and by Claim \ref{claim:bounding-the-potential2} we know that 
\[
\Phi_k(\TT) + \dist_{\TT}(l_\ALG, l_{\OPT_{\TT}}) \leq D_k \OPT_{\TT}.
\]

By Claim \ref{claim:extreme-procedure}, $\OPT_{\TT^{new}} = \OPT_{\TT}$ and $\Phi_{k+1}(\TT^{new}) = D_{k+1} \OPT_{\TT} > D_k \OPT_{\TT}$, so we get that $\Phi(\TT^{new}) > \Phi(\TT) + \Delta \cost(\ALG)$, as desired.

Otherwise, if the fork does not cause an increase of $k$, it is easy to see that $\cost(\ALG)$ and $\Phi(T)$ stay the same, and the ratio invariant remains satisfied.

\subsubsection{Bounding the distortion factor}
\label{sec:bounding-C}

In this section we bound the distortion in $\TT$.

\begin{claim}
    The distortion factor of a (non-zero length) edge $e$ at level $i$ satisfies
    \[1 \leq \frac{w_e}{w^0_e} \leq \Pi_{j=2}^i {x_{j+1} \dots x_k}.\]
    In particular, an edge $e$ at level $1$ is not distorted at all (i.e.\ $w_e = w^0_e$).
    \label{claim:bounding-C}
\end{claim}
\begin{proof}
    First note that growing $e$ cannot increase ${w_e}/{w^0_e}$, since $w_e$ and $w^0_e$ grow by the same amount.  
    
    It remains to bound how much $e$ is scaled in the imbalancing operations. Note that edges at level $1$ are never scaled: Edges scaled upon a deletion were promoted just before, so their level is greater than $1$; if a fork step causes scaling, then the only level 1 edges are the two new edges of length $0$, so scaling them has no effect. Thus, assume that $i>1$. When we apply the extreme-imbalance procedure to a subtree $S$ of level $j>i$ which contains $e$, $e$ is scaled by at most $x_j$. Considering all of the recursive calls of the procedure, $e$ is scaled by at most \[x_j x_{j-1} \dots x_{i+1} \leq x_k x_{k-1} \dots x_{i+1}.\] Since we apply the extreme-imbalance procedure only to subtrees whose level has just been increased, the level of $e$ was at most $i-1$ when the previous scalings occurred, and we can conclude by an inductive argument.
    \end{proof}

\begin{claim}
    Let $C := \max_e \frac{w_e}{w^0_e}$ be the distortion factor of $\TT$. Then $C < 60$.
\end{claim}
\begin{proof}
By Claim \ref{claim:bounding-C}, 
\[
C \leq (x_3 \dots x_k) \dots (x_{k-1} x_k) x_k = \Pi_{i=3}^k x_i^{i-2}.
\]

Since $x_i = 1 + \sqrt{\frac{2}{1 + D_{i-1}}}$ and $D_{i-1} > 2^{i+1}$ for $i \geq 5$, we have $x_i < 1 + 2^{-i/2}$. Using that $1 + r \leq e^r$ for $r \in \reals$, we get $\ln(x_i) < 2^{-i/2}$. Thus,
\begin{align*}
    \ln(C) &\le \sum_{i=3}^k (i-2) \ln(x_i) 
            < \ln(x_3) + 2\ln(x_4) + \sum_{i=5}^{\infty} (i-2) 2^{-i/2} 
            \approx 4.09,
\end{align*}
so $C < 60$.
\end{proof}

\subsection*{Acknowledgments}
C. Coester is funded by the European Union (ERC, CCOO, 101165139). Views and opinions expressed are however those of the author(s) only and do not necessarily reflect those of the European Union or the European Research Council. Neither the European Union nor the granting authority can be held responsible for them.

\bibliographystyle{alpha}
\bibliography{ref}

\newpage
\appendix

\section{\texorpdfstring{Tight bound of $D_3$ for width $3$}{Tight bound of D₃ for width 3}}\label{app:D3}
When the game is played on trees with at most $3$ leaves, we can match the lower bound precisely. We prove this by refining both the algorithm and the potential function used to analyze it. 

\subsection{Algorithm description}
Our refined algorithm does not use the imbalancing operation, and works directly on the tree from the evolving tree game, which we denote by $T$ in this section. The algorithm responds to growth in the same way as in Section \ref{sec:algo-description}. For deletion, the algorithm uses a very natural strategy: it remains in its current position if the ratio invariant remains satisfied, and otherwise moves to the other leaf (note that there will be at most two leaves after a deletion). For forking, the algorithm attempts to remain in its location (or, if the fork occurs there, move to one of the new leaves). If this strategy makes the ratio invariant become violated, the algorithm moves to the optimal leaf. Observe that the ratio invariant is maintained.

\subsection{Refining the potential}
Consider the tree $T$ from Fig.\ \ref{fig:three-leaves}, where without loss of generality $LT$ contains only one leaf, and $RT$ contains two leaves. In the previous version of the potential, we defined \begin{equation}
\overline{\Phi(LT)} := D_2(a \land x_3 \OPT_{RT}) = D_2\left(a \land x_3 \left(b + (c \land d) \right)\right).    \label{eq:phibar-left-old}
\end{equation} We refine this to \begin{equation}
\overline{\Phi(LT)} := D_2 \left(a \land bx_3\right) + D_1 \left[(a - b x_3) \land ((c\land d) x_3)\right]_+,
\label{eq:phibar-left}
\end{equation} where $[r]_+ := \max(r, 0)$ for any $r \in \reals$.

The intuition is as follows: Just before the last fork (which gave rise to $l_2$ and $l_3$) occurred, there were only two edges in the tree, of lengths $a' \leq a$ and $b$. Up to that moment, the algorithm had paid at most $D_2 (a' \land b x_3)$ while staying in $LT$, as $LT$ might have had up to two leaves. After the fork, the algorithm paid cost at most $D_1 ((a \land x_3 \OPT_{RT}) - a')$ while staying in $LT$, as $LT$ must have had only one leaf. Maximizing over $a' \in [0, a]$, and using $\OPT_{RT} = b + (c \land d)$, we get that the total cost paid by the algorithm while staying in $LT$ is bounded by \eqref{eq:phibar-left}.

If $RT$ contains only one leaf, we substitute $c = d = 0$ in \eqref{eq:phibar-left} to obtain the same expression for $\overline{\Phi(LT)}$ as we previously had from \eqref{eq:phibar-left-old}. We define $\Phi(T)$ and $\overline{\Phi(RT)}$ as before, using \eqref{eq:updated-potential}, \eqref{eq:opt-other-cap}, and \eqref{eq:phi-cap}. 

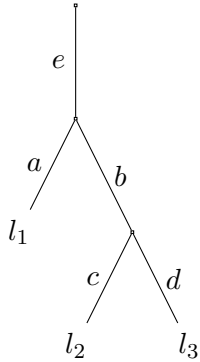
\begin{figure}[ht]
    \centering
\begin{tikzpicture}
\node[draw, inner sep=0pt] {} 
    child {node[draw, inner sep=0pt] {} 
        child {node {$l_1$} 
        edge from parent node [left] {$a$}}
        child {node[draw, inner sep=0pt] {} 
            child {node {$l_2$} edge from parent node [left] {$c$}}
            child {node {$l_3$} edge from parent node [right] {$d$}}
        edge from parent node [right] {$b$}}
    edge from parent node [left] {$e$}};
\end{tikzpicture}
\caption{A tree $T$ with three leaves $\{l_1, l_2, l_3\}$. The edges are labelled by their weights.}
\label{fig:three-leaves}
\end{figure}

\subsection{Analysis}
It is clear that the refined potential does not exceed the previous version of the potential, so Claims \ref{claim:bounding-the-potential} and \ref{claim:bounding-the-potential2} still hold. The previous analysis for the growth and fork operations from sections \ref{sec:growth-analysis} and \ref{sec:forking-analysis} can easily be adapted to work with the refined potential. We sketch proofs below, omitting straightforward calculations.

\begin{claim}
    We have $\Delta \cost(\ALG) \leq \Delta \Phi(T)$ for growth.
\end{claim}
\begin{proof}
 Suppose $T$ has three leaves, and let the edge weights be as in Fig. \ref{fig:three-leaves}. If $\ALG$ does not move after the growth operation, it is easy to see that potential term corresponding to the growing leaf increases precisely by $\Delta \cost(\ALG)$. Otherwise, if $\ALG$ moves, we can reuse the previous analysis to show that $\overline{\OPTOTHER(T)} + \overline{\Phi(RT)}$ increases enough to pay for the movement. If $T$ has one or two leaves, we obtain the same expression for $\Phi(T)$ as in Section \ref{sec:updated-potential}, and therefore we can reuse the previous analysis.
\end{proof}

\begin{claim}
    We have $\Delta \cost(\ALG) \leq \Delta \Phi(T)$ for fork.
\end{claim}
\begin{proof}
Recall that the algorithm remains in its location (or, if the fork occurs there, moves to one of the new leaves), provided that this does not violate the ratio invariant. Notice that the only scenario in which this strategy would violate the ratio invariant is if the fork triggers an increment of $k$ from $2$ to $3$, the ratio between $w_{e(LT)}$ and $w_{e(RT)}$ is more than $x_3$, and the algorithm is located in the non-optimal leaf $l_\ALG$ of $T$. In this case, the tree $T'$ obtained after the operation is extreme, so we have $\Phi_3(T') = D_3 \OPT_{T'}$. On the other hand, by Claim \ref{claim:bounding-the-potential2}, we have $\Phi_2(T) + \dist(l_\ALG, \OPT_T) \leq D_2 \OPT_T$, so we obtain $\Phi_3(T') - \Phi_2(T) \geq \Delta \cost(\ALG)$, as desired. Otherwise, if the ratio invariant is not violated, $\Delta \Phi(T) \geq 0 = \Delta \cost(\ALG)$.
\end{proof}

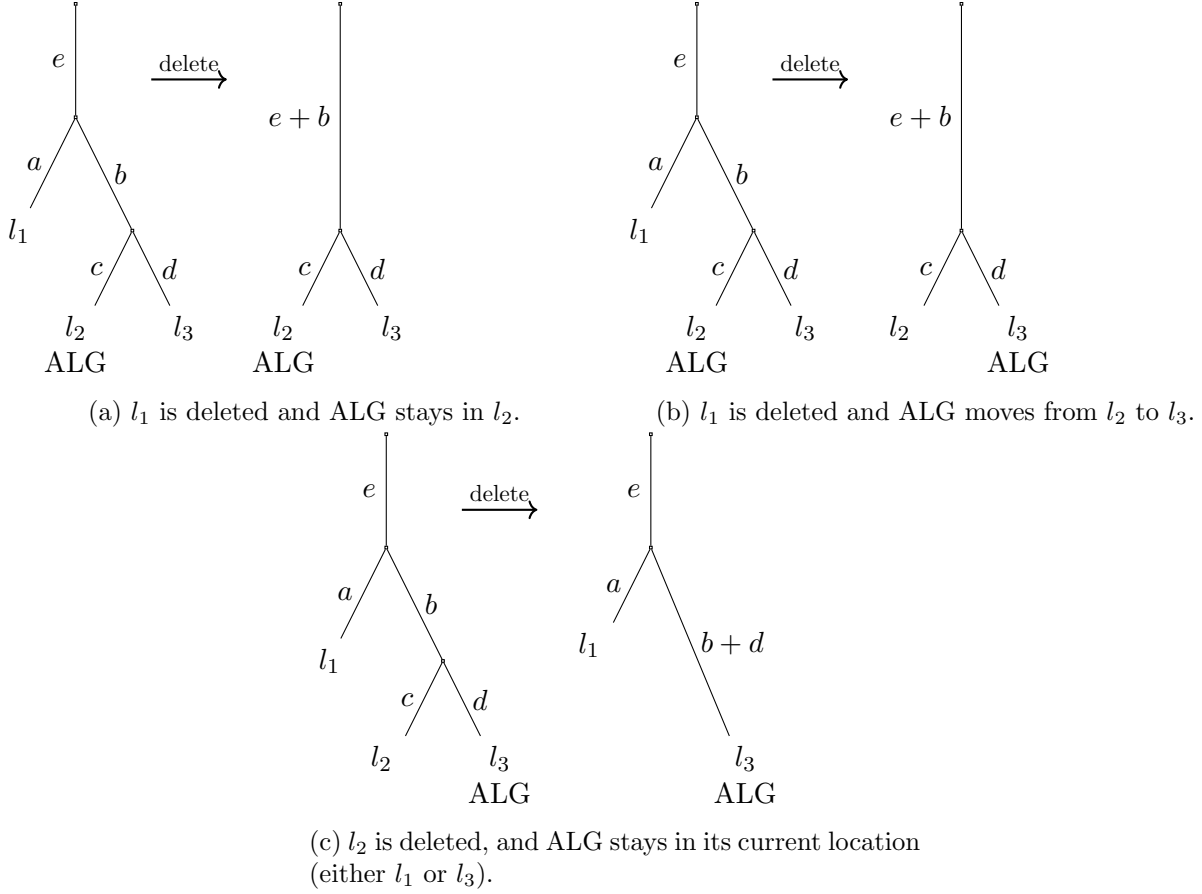
\begin{figure}[h]
\centering
\begin{subfigure}[t]{0.49\textwidth}
\begin{tikzpicture}
\node[draw, inner sep=0pt] (x) at (0,0) {} 
    child {node[draw, inner sep=0pt] {} 
        child {node {$l_1$} 
        edge from parent node [left] {$a$}}
        child {node[draw, inner sep=0pt] {} 
            child {node[text width=1cm,align=center]  {$l_2$ \\ $\ALG$} edge from parent node [left] {$c$}}
            child {node[text width=1cm,align=center]  {$l_3$ \break} edge from parent node [right] {$d$}}
        edge from parent node [right] {$b$}}
    edge from parent node [left] {$e$}};

\node[draw, inner sep=0pt] (y) at (3.5, 0) {} 
    child {node[draw, inner sep=0pt] at (0, -1.5) {} 
        child {node[text width=1cm,align=center]  {$l_2$ \\ $\ALG$}
        edge from parent node [left] {$c$}}
        child {node[text width=1cm,align=center]  {$l_3$ \break}
        edge from parent node [right] {$d$}}
    edge from parent node [left] {$e+b$}};

\draw[->, thick] (1, -1) -- (2, -1) node[midway, above, scale=0.8] {delete};

\end{tikzpicture}
\caption{$l_1$ is deleted and $\ALG$ stays in $l_2$.}
\end{subfigure}
\begin{subfigure}[t]{0.49\textwidth}
\begin{tikzpicture}
\node[draw, inner sep=0pt] (x) at (0,0) {} 
    child {node[draw, inner sep=0pt] {} 
        child {node {$l_1$}
        edge from parent node [left] {$a$}}
        child {node[draw, inner sep=0pt] {} 
            child {node[text width=1cm,align=center]  {$l_2$ \\ $\ALG$} edge from parent node [left] {$c$}}
            child {node[text width=1cm,align=center]  {$l_3$ \break} edge from parent node [right] {$d$}}
        edge from parent node [right] {$b$}}
    edge from parent node [left] {$e$}};

\node[draw, inner sep=0pt] (y) at (3.5, 0) {} 
    child {node[draw, inner sep=0pt] at (0, -1.5) {} 
        child {node[text width=1cm,align=center]  {$l_2$ \break} 
        edge from parent node [left] {$c$}}
        child {node[text width=1cm,align=center]  {$l_3$ \\ $\ALG$} 
        edge from parent node [right] {$d$}}
    edge from parent node [left] {$e+b$}};

\draw[->, thick] (1, -1) -- (2, -1) node[midway, above, scale=0.8] {delete};
\end{tikzpicture}

\caption{$l_1$ is deleted and $\ALG$ moves from $l_2$ to $l_3$.}
\end{subfigure}
\begin{subfigure}[t]{0.49\textwidth}
\begin{tikzpicture}
\node[draw, inner sep=0pt] (x) at (0,0) {} 
    child {node[draw, inner sep=0pt] {} 
        child {node {$l_1$}
        edge from parent node [left] {$a$}}
        child {node[draw, inner sep=0pt] {} 
            child {node[text width=1cm,align=center]  {$l_2$ \break} edge from parent node [left] {$c$}}
            child {node[text width=1cm,align=center] {$l_3$ \\ $\ALG$}  edge from parent node [right] {$d$}}
        edge from parent node [right] {$b$}}
    edge from parent node [left] {$e$}};

\node[draw, inner sep=0pt] (y) at (3.5,0) {}
    child {node[draw, inner sep=0pt] {} 
        child {node[text width=1cm,align=center]  {$l_1$ \break}
        edge from parent node [left] {$a$}}
        child {node[text width=1cm,align=center] at (0.5, -1.5)  {$l_3$ \\ $\ALG$} 
        edge from parent node [right] {$b+d$}}
    edge from parent node [left] {$e$}};

\draw[->, thick] (1, -1) -- (2, -1) node[midway, above, scale=0.8] {delete};
\end{tikzpicture}
\caption{$l_2$ is deleted, and $\ALG$ stays in its current location (either $l_1$ or $l_3$).}
\end{subfigure}

\caption{Deletion cases. $\ALG$ represents the position of the algorithm and the labels next to the edges indicate weights.}
\label{fig:deletion-for-three-leaves}
\end{figure}

\begin{claim}
    We have $\Delta \cost(\ALG) \leq \Delta \Phi(T)$ for deletion.
\end{claim}
\begin{proof}
We assume without loss of generality that the edge adjacent to the leaf to be deleted is large enough (otherwise we could grow that edge before doing the deletion). By the ratio invariant, this ensures that $\ALG$ is not located at that leaf. Let $T'$ be the tree obtained after deletion. If $T$ has only two leaves before deletion, $\ALG$ does not move and hence $\Delta \cost(\ALG) = 0$, and $\Phi(T) \leq \Phi(T') = D_3 \OPT_{T}$. Thus, assume that $T$ has three leaves before deletion. We split our analysis into several cases, as shown in Fig. \ref{fig:deletion-for-three-leaves}.

    \paragraph{(a): $l_1$ is deleted and $\ALG$ does not move.}
    Since $\ALG$ does not move, we have $\Delta \cost(\ALG) = 0$. Suppose without loss of generality that $\ALG$ is initially in $l_2$. If $d \geq x_3 c$, then $T'$ is extreme, so $\Phi(T') = D_3 \OPT_{T'} = D_3 \OPT_T$. By Claim \ref{claim:bounding-the-potential}, $\Phi(T) \leq D_3 \OPT_{T}$, so $\Delta \Phi(T) \geq 0$ and we are done. Thus, suppose that $d < x_3 c$. Since $\ALG$ does not move to $l_3$ after deletion, we also have $c \leq x_3 d$. Note that $\Phi(T)$ is increasing in $a$ and the maximum is achieved when $a \geq (b + (c\land d))x_3$, so 
    \begin{align*}
    \Phi(T) &\leq D_3 e + D_2 bx_3 + D_1 (c \land d) x_3 + D_2 b + D_1 c + D_1 d + \frac{x_3+1}{x_3-1} (b + (c\land d))x_3 + \frac{x_2+1}{x_2-1} d.
    \end{align*}

    Note that $b$ has coefficient $(x_3+1)D_2 + \frac{x_3(x_3+1)}{x_3-1} = D_3$, so the expression simplifies to
    \begin{align*}
    \Phi(T) &\leq D_3 (e+b) + D_1 x_3 (c \land d) + D_1 c + D_1 d + \frac{x_3(x_3+1)}{x_3-1} (c\land d) + \frac{x_2+1}{x_2-1} d.
    \end{align*}

    After deletion,
    \[\Phi(T') = D_3 (e+b) + D_2 c + D_2 d + \frac{x_3+1}{x_3-1} d.\]

Since $c \land d \leq d$, 
\begin{align*}
    \Phi(T') - \Phi(T) &\geq (D_2 - D_1) c +  \left(D_2 - D_1 + \frac{x_3+1}{x_3-1} - \frac{x_2+1}{x_2-1} - D_1 x_3 - \frac{x_3(x_3+1)}{x_3-1}\right) d \\
    &= (D_2 - D_1) c +  \left(D_2 - D_1 (1+x_3) - \frac{x_2+1}{x_2-1} - (x_3+1)\right) d \\
    &\geq (D_2 - D_1) c +  \left(D_2 - D_1 (1+x_2) - \frac{x_2+1}{x_2-1} - (x_2+1)\right) d \\
    &= (D_2 - D_1) c \\
    &\geq 0 = \Delta \cost(\ALG).
\end{align*}

In the last two steps, we used the fact that $x_3 < x_2$ and we plugged in the formula for $D_2$, which makes the term involving $d$ vanish.

  \paragraph{(b): $l_1$ is deleted and $\ALG$ moves.}
    Suppose without loss of generality that $\ALG$ was in $l_2$ and it moves to $l_3$ after deletion, paying cost $\Delta \cost(\ALG) = \dist_T(l_2, l_3)$. The move is triggered because the ratio invariant becomes unsatisfied, so we have $c > x_3 d$, which means that $T'$ is extreme. Therefore, \[\Phi(T') = D_3 \OPT_{T'} = D_3 \OPT_T.\] By Claim \ref{claim:bounding-the-potential2}, \[\Phi(T) + \dist_T(l_2, l_3) \leq D_3 \OPT_T.\] Therefore, $\Delta \cost(\ALG) \leq \Phi(T') - \Phi(T)$.

\paragraph{(c): $l_2$ is deleted.}
First note that $\OPT_{RT'} = \OPT_{RT} = b + d$, since we assumed that the leaf $l_2$ grows by a sufficiently large amount before it is deleted, so $c>d$. Since $\OPT_{LT'} = \OPT_{LT}$ and $\OPT_{RT'} = \OPT_{RT}$, it is clear that the ratio invariant remains satisfied after deletion, so $\ALG$ does not move and $\Delta \cost(\ALG) = 0$. Also, $\Phi(RT') = D_2 \OPT_{RT} \geq \Phi(RT)$ by Claim \ref{claim:bounding-the-potential}, so $\Phi(T') \geq \Phi(T)$.    

The case when $l_3$ is deleted is symmetrical to the above.
\end{proof}

\section{\texorpdfstring{Why a tight bound of $D_k$ for width $k\geq 4$ is difficult}{Why a tight bound of Dₖ for width k is difficult}}\label{app:towards-D4}

In this section, we consider a natural extension of the approach used in Appendix \ref{app:D3} to prove an upper bound of $D_3$ for width $3$, and show that this fails for $k \geq 4$. We hope that this provides some intuition about the difficulties that arise in the general case. We consider the same algorithm, which maintains the ratio invariant. As in the case for width $k=3$, we seek to identify subtrees for which the algorithm must have paid less than suggested by the level of the subtree, and refine the potential accordingly. 

Consider the example for width $k=4$ in Fig. \ref{fig:deletion-examples-width4}. After the fork in $RT$ happens, the number of possible leaves in $LT$ is limited to $2$. Since the algorithm maintains the ratio invariant, it does not reach $LLT$ or $RLT$ before the fork in $RT$ happens, and therefore the algorithm visits $LLT$ and $RLT$ only when they are trivial subtrees. Thus, the potentials corresponding to
$LLT$ and $RLT$ should be refined to use a rate of $D_1$ rather than $D_2$. With this refinement, we obtain 
\[
\Phi(LT) = D_3 ax_4 + D_1 b + D_1 bx_3 + \frac{x_3+1}{x_3-1} bx_3. 
\]
After deleting leaf $l$, the refined potential becomes
\[
\Phi(LT') = D_3 ax_4 + D_2 b.
\]
We have
\begin{align*}
\Delta \Phi(LT) &= \Phi(LT') - \Phi(LT) 
= \left(D_2 - (x_3+1) D_1  - \frac{x_3(x_3+1)}{x_3-1}\right)b < 0.
\end{align*}

Observe that the potential decreases after a deletion operation, which means that this refined potential cannot be used to prove a tight upper bound of $D_4$. Intuitively, the problem is that refining $\Phi(LT')$ by noting that the algorithm’s cost increases at rate $D_2$ on the segment of length $b$, rather than at rate $D_3$, makes $\Phi(LT')$ too small. The corresponding refinement of $\Phi(LT)$ does not reduce $\Phi(LT)$ sufficiently to offset this effect and avoid a decrease of the potential.

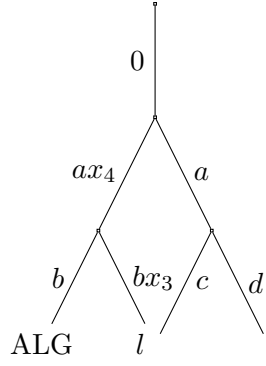
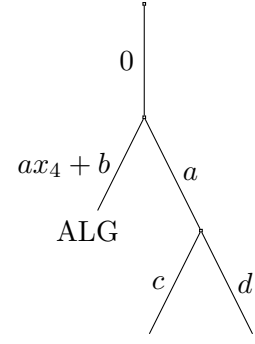
\begin{figure}[h]
\centering
\begin{subfigure}[ht]{0.49\textwidth}
\centering
\begin{tikzpicture}
\node[draw, inner sep=0pt] {}
    child { 
        node[draw, inner sep=0pt] {}
    child {
        node[draw, inner sep=0pt] {}
            child {
                node {$\ALG$}
                edge from parent node[left] {$b$}
            }
            child {
                node {$l\quad$}
                edge from parent node[right] {$bx_3$}
            }
        edge from parent node[left] {$ax_4$}
    }
    child {
        node[draw, inner sep=0pt] {}
            child {
                node {}
                edge from parent node[right] {$c$}
            }
            child {
                node {}
                edge from parent node[right] {$d$}
            }
        edge from parent node[right] {$a$}
    }
    edge from parent node[left] {$0$}
    };
\end{tikzpicture}
\caption{Tree $T$.}
\label{fig:width4}
\end{subfigure}
\begin{subfigure}[ht]{0.49\textwidth}
\centering
\begin{tikzpicture}
\node[draw, inner sep=0pt] {} 
    child {node[draw, inner sep=0pt] {} 
        child {node {$\ALG$} 
        edge from parent node [left] {$ax_4+b$}}
        child {node[draw, inner sep=0pt] {} 
            child {node {} edge from parent node [left] {$c$}}
            child {node {} edge from parent node [right] {$d$}}
        edge from parent node [right] {$a$}}
    edge from parent node [left] {$0$}};
\end{tikzpicture}
\caption{Tree $T'$ obtained when deleting leaf $l$ from $T$.}
\label{fig:width4-deletion}
\end{subfigure}
\caption{Deletion example for width $k=4$. $\ALG$ represents the position of the algorithm and the labels next to the edges indicate weights.}
\label{fig:deletion-examples-width4}
\end{figure}

\section{\texorpdfstring{Proof that $D_k = O(2^k)$}{Proof that Dₖ = O(2ᵏ)}}
\label{sec:bounding-D}
\begin{claim}
For $k\geq 2$, it holds that 
\[D_k \leq 2^{k+4} - \sqrt{2^{k+9}}.\]
\end{claim}
\begin{proof}
By induction. The base case $k=2$ holds because $D_2 = 9 < 2^6 - \sqrt{2^{11}}$. For the inductive step, we plug in the formula for $D_{k+1}$ and use the inductive hypothesis to obtain
\begin{align*}
    D_{k+1} &= 2D_k + \sqrt{8(1+D_k)} + 3 \leq 2^{k+5} - 2\sqrt{2^{k+9}} + \sqrt{8 \cdot 2^{k+4}} + 3.
\end{align*}
To complete the proof, it suffices to show that 
\[
2\sqrt{2^{k+9}} \geq \sqrt{8 \cdot 2^{k+4}} + 3 + \sqrt{2^{k+10}}.
\]
This is equivalent to 
\[
\sqrt{2^{k+7}} (3 - 2 \sqrt{2}) \geq 3,
\]
which holds for all $k \geq 2$.
\end{proof}

\section{Lower bound for deterministic algorithms}\label{app:LB}

We now give a rigorous proof of Theorem~\ref{th:lower-bound}, expanding on the sketch given in Section~\ref{sec:LB}.

\paragraph{Proof outline.} We show that for every $k\in\mathbb N$, every $\epsilon > 0$, there exists $\delta\in(0,1]$ such that for every $L\ge 0$ and every deterministic online algorithm $A$, there exists a layered graph traversal instance of width $k$ with a single node in the last layer at some distance $\ell\in[\delta\cdot L, L]$ from the source, and where $A$ moves at least distance $(D_k-\epsilon)\cdot\ell$.

The proof is by induction on $k$. For $k\ge 2$, the constructed tree has of two branches, and we always extend the branch where the algorithm is currently located with an instance of width $k-1$ as implied by the induction hypothesis, while the other branch is extended with edges of length $0$. The upper bound $L$ on the length of the instances allows us to ensure that if the algorithm switches to the other branch in the middle of a recursive instance, the cost saved through this is negligible compared to the cost of switching to the other branch. On the other hand, the lower bound $\delta\cdot L$ on the length of recursive instances allows us to create instances of arbitrary longer length by concatenating finitely many shorter instances.

\paragraph{Base of induction.} The base of induction $k=1$ is trivial, using a single edge of length $L$.

\paragraph{Induction step.} For $k\ge 2$, we construct a hard instance as follows. Let $1\gg\epsilon'\gg\delta>0$
be small constants, as determined later, with $1/\epsilon'\in\mathbb N$. Let $\delta'\in(0,1]$ be the constant induced by the induction hypothesis applied to $k-1$ and $\epsilon'$. Let $L'=L\delta\epsilon'$.

Layer $0$ contains only the source, and layer $1$ contains two nodes, each connected to the source with an edge of length $\delta L$. These two nodes will constitute the initial nodes of two branches. At any time, we call the \emph{length} of a branch the length of a shortest path from the source to the most recent layer of the branch. 

We construct the instance in phases, where each phase has the following structure: At the start of the phase, the current layer contains only a single node per branch. We call the branch where the algorithm is located at the start of the phase \emph{active}, and the other branch \emph{passive}. During the phase, each new layer contains a single node on the passive branch, connected to the previous node with an edge of length $0$. More importantly, each layer during the phase contains up to $k-1$ nodes on the active branch. To specify the nodes on the active branch, consider the algorithm $A'$ for instances of width $k-1$ induced by the behavior of $A$ on the active branch during the phase. Note that $A'$ may not be well-defined since $A$ might switch to the passive branch at some point. If this happens, we define the subsequent behavior of $A'$ arbitrarily. By the induction hypothesis applied to $k-1$, $\epsilon'$, $L'$ and $A'$, we can choose nodes and edges on the active branch such that the length of the active branch increases by some $\ell\in[\delta'L',L']$ during the phase, and $A$ pays at least $(D_{k-1}-\epsilon')\ell$ if it stays in the active branch during the phase. 

We call a \emph{super-phase} a maximal consecutive sequence of phases during which the same branch is active. Denote by $\opt_t$ the length of the \emph{passive} branch during super-phase $t$; this quantity is constant during the super-phase, since the passive branch is only extended with $0$-length edges. Moreover, let $\opt_0=\delta L$ be the length of the branches before the first super-phase. If $\opt_t\le L$ and the active branch grows to length at least $D_k\opt_t$ during super-phase $t$, then we can easily finish the instance by adding one more layer with a single node connected via a length-$0$ edge to the last node of the passive branch. Thus, we may assume that
\begin{align}
    \forall t\colon \opt_t\le L\implies \opt_{t+1}< D_k \cdot\opt_t.\label{eq:slowGrowthBeforeL}
\end{align}
Let
\begin{align*}
    t_0 = \min\{t\ge 1\mid \opt_t\ge \delta L/\epsilon'\text{ for all }t\ge t_0\}
\end{align*}
be the identifier of the first super-phase after both branches reached length $\delta L/\epsilon'$. The following claim establishes lower and upper bounds on $t_0$. The precise bounds are not important to us, but we will need that they are unbounded increasing functions of $1/\epsilon'$ and independent of $\delta$ (except $\delta$ needs to be sufficiently small).
\begin{claim}\label{cl:t0Bounds}
    If $\delta$ is sufficiently small, then $\frac{\log\frac{1}{\epsilon'}}{\log D_k}< t_0\le \frac{1+D_k}{\delta'(\epsilon')^2}$.
\end{claim}
\begin{proof}
    We first show that
    \begin{align}
        \opt_{t}\le D_k\cdot \opt_{t-1}\text{ for all }t\le t_0.\label{eq:slowGrowthBeforet0}
    \end{align}
    If $\opt_{t-1}\le L$, then this follows immediately from \eqref{eq:slowGrowthBeforeL}. If $t$ and $t_0$ have the same parity, then this is the case since the same branch is passive during super-phases $t-1$ and $t_0-1$, hence $\opt_{t-1}\le \opt_{t_0-1}\le \delta L/\epsilon'\le L$ for $\delta$ sufficiently small. If instead $t$ and $t_0$ have opposite parities and $\opt_{t-1}> L$, then for sufficiently small $\delta$ we get $D_k\opt_{t-1}>D_kL\ge \frac{\delta L}{\epsilon'}> \opt_{t_0-1}\ge\opt_{t}$, where the last inequality uses that the same branch is passive in super-phases $t$ and $t_0-1$, by parity. This proves \eqref{eq:slowGrowthBeforet0}.
    
    From \eqref{eq:slowGrowthBeforet0} and $\opt_0=\delta L$, we get that $\opt_{t_0}\le D_k^{t_0}\delta L$, which together with $\opt_{t_0}\ge \delta L/\epsilon'$ implies that $t_0\ge \frac{\log\frac{1}{\epsilon'}}{\log D_k}$. This establishes the lower bound on $t_0$.
    
    On the other hand, \eqref{eq:slowGrowthBeforet0} yields $\opt_{t_0-1}+\opt_{t_0} \le (1+D_k)\opt_{t_0-1} \le (1+D_k)\delta L/\epsilon'$. But the active branch grows by at least $\delta'L'=\delta L\delta'\epsilon'$ per phase (and super-phase), so $\opt_{t_0-1}+\opt_{t_0}\ge t_0\delta L\delta'\epsilon'$ and hence $t_0\le \frac{1+D_k}{\delta'(\epsilon')^2}$.
\end{proof}

The instance ends when the imbalance $\frac{\opt_{t+1}}{\opt_t}$ between the two branch lengths is close to the extreme imbalance that was observed throughout the instance, except possibly during the initial $\epsilon'$ fraction of super-phases. More precisely, the last super-phase is super-phase $T$, where
\begin{align*}
    T=\inf\left\{ T\ge t_0\,\middle\vert\, \frac{\opt_{T+1}}{\opt_T} \ge \max_{t\in[\epsilon' T,T)}\frac{\opt_{t+1}}{\opt_t} -  \epsilon'\right\}.
\end{align*}

\begin{claim}
    It holds that $T<\infty$ and $\opt_T\le L$.
\end{claim}
\begin{proof}
    Let $\tilT$ be minimal such that $\opt_{\tilT+1}>L$. Such $\tilT$ exists, since any two consecutive super-phases increase both branch lengths by at least $\delta'L'$. To prove the claim, it suffices to show that $\tilT\ge T$. Suppose towards a contradiction that $\tilT< T$.
    
    Let $i\in\mathbb Z_{\ge 0}$ be such that $\tilT\in\left[t_0/(\epsilon')^{i},t_0/(\epsilon')^{i+1}\right)$. Then $\frac{\opt_{\tilT+1}}{\opt_{\tilT}}\le D_k-i\epsilon'$ by induction on $i$, where the base of induction is due to~\eqref{eq:slowGrowthBeforeL}. But then $i\le D_k/\epsilon'$, hence $\tilT< t_0/(\epsilon')^{D_k/\epsilon'+1}\le (1+D_k)/(\delta'(\epsilon')^{D_k/\epsilon'+3})$ using Claim~\ref{cl:t0Bounds}. Thus,
    \begin{align}
        \opt_{\tilT+1}\le \opt_{t_0-1}\cdot D_k^{\tilT+2-t_0}\le  \delta L/\epsilon'\cdot D_k^{(1+D_k)/(\delta'(\epsilon')^{D_k/\epsilon'+3})},\label{eq:opttTBound}
    \end{align}
    where the first inequality uses \eqref{eq:slowGrowthBeforeL}, and the second inequality uses $\opt_{t_0-1}\le\delta LD_k/\epsilon$ and the upper bound on $\tilT$. But for $\delta$ sufficiently small, the right-hand side of \eqref{eq:opttTBound} is at most $L$, which yields the desired contradiction.
\end{proof}

After super-phase $T$, we add one more layer with a single node connected by a length-$0$ edge to the last node of the passive branch of super-phase $T$. Thus, the shortest path from the source to this node has length $\opt_T\in[\delta\cdot L,L]$, as required.

\paragraph{Online cost analysis.} We now analyze the cost of $A$ on this instance. In super-phase $t$, the active branch grows from length $\opt_{t-1}$ to length $\opt_{t+1}$, so $A$ pays at least $(D_{k-1}-\epsilon')(\opt_{t+1}-\opt_{t-1}-L')$ for advancing on this branch during the super-phase, since it can reside in the passive branch only during part of the last phase of the super-phase where the growth is at most $L'$. Additionally, switching to the passive branch costs at least $\opt_{t+1}+\opt_t-L'$, of which $\opt_{t+1}-L'$ is for backtracking to the source and $\opt_t$ is for advancing to the leaf in the passive branch. Thus, the total cost of $A$ in super-phase $t$ is at least
\begin{align*}
    (D_{k-1}-\epsilon')(\opt_{t+1}-\opt_{t-1}-L') + \opt_{t+1}+\opt_t-L'.
\end{align*}
Letting $x=\max_{t\in[\epsilon' T,T)}\frac{\opt_{t+1}}{\opt_t}> 1$ and writing ``$\gtrsim$'' for inequalities that hold up to a multiplicative factor that can be made arbitrarily close to $1$ by choosing $\epsilon'$ sufficiently small, the total cost of $A$ is at least
\begin{align*}
    &\sum_{t=1}^T\left[(D_{k-1}-\epsilon')(\opt_{t+1}-\opt_{t-1}-L') + \opt_{t+1}+\opt_t-L'\right]\\
    &\gtrsim\sum_{t=1}^T\left[D_{k-1}(\opt_{t+1}-\opt_{t-1}) + \opt_{t+1}+\opt_t\right]\\
    &= D_{k-1}\left(\opt_{T+1}+\opt_T-\opt_1-\opt_0\right) + \opt_{T+1}-\opt_1 + 2\sum_{t=1}^T\opt_t\\
    &\gtrsim \opt_T\cdot\left(D_{k-1}(1+x)+x + 2\sum_{t=0}^{T-1}\frac{1}{x^t}\right),
\end{align*}
where the first inequality uses $(D_{k-1}+1)L'=(D_{k-1}+1)\epsilon'\delta L\ll\delta L\le \opt_{t+1}$, and the last inequality  uses $\opt_{T+1}\ge (x-\epsilon')\cdot \opt_T\gtrsim x \cdot \opt_T$ and $\opt_0=\opt_1=\delta L\ll \delta L/\epsilon'\le \opt_T$ since $T\ge t_0$.

By the lower bound on $t_0$ in Claim~\ref{cl:t0Bounds}, choosing $\epsilon'$ small makes $T\ge t_0$ arbitrarily large. Thus, the proof of the induction step is completed if
\begin{align*}
    D_k&= \min_{x>1}D_{k-1}(1+x)+x + 2\sum_{t=0}^{\infty}\frac{1}{x^t} \\
    &= \min_{x>1}D_{k-1}(1+x)+\frac{x(x+1)}{x-1}.
\end{align*}
The minimum is attained for $x=1+\sqrt{\frac{2}{1+D_{k-1}}}=x_k$. Thus, the Lemma~\ref{lem:DkViaxk} completes the proof.

\section{Lower bound against adaptive adversaries} \label{sec:LB-adaptive}

This section is dedicated to proving Theorem \ref{th:adaptive-adversaries}.
In the \emph{adaptive online adversary} model, requests are issued by an adversary that can observe the past decisions made by the online algorithm. When issuing a new request, the adversary serves it immediately (which means that adversary's decisions might turn out to be sub-optimal once the full instance is revealed). The performance of an algorithm is defined as the ratio between its expected cost and the expected cost incurred by the adversary. For the stronger adaptive \emph{offline} adversary (which serves the instance optimally in hindsight), our $D_k$ lower bound translates immediately since randomization is known to provide no advantage in that setting~\cite{Ben-DavidBKTW94}.

Our proof is adapted from the lower bound on deterministic algorithms in~\cite{Fiat-competitiveLGT}. Adapting our lower bound of $D_k$ from Appendix \ref{app:LB} seems challenging, as the instance can only be stopped at certain times when the ratio between the branch
lengths is maximized, and the adversary does not know in advance which branch will be optimal at those times.

We construct our lower bound instance for layered graph traversal as follows. Let $\ALG$ be a (possibly randomized) online algorithm for this problem. As soon as the next layer of the instance is revealed, we will also choose the next move of the adversary $\ADV$ to a node in this layer. The instance is constructed inductively on $k$. We start by connecting the source $s$ to two nodes $l_1$ and $r_1$ via edges of length $C_k$ (which will be determined later). These nodes constitute the initial nodes of two branches $L$ and $R$. We will define $\ADV_L$ and $\ADV_R$ to be adversaries running in $L$ and $R$, respectively. The adversary $\ADV$ chooses one branch $X \in \{L, R\}$ at random (each having equal probability), moves to the initial node in that branch, and then follows the movement of $\ADV_X$ for the rest of the instance. 
    
    Without loss of generality, suppose that $\ALG$'s first move is to advance to $l_1$. Then, we play a sub-instance of width $k-1$ rooted at $l_1$. During this time, $\ADV_L$ follows the movement of the adversary in this sub-instance. While the sub-instance is played in $L$, we extend $R$ by a path of zero-length edges, and we let $\ADV_R$ follow this path. We stop the sub-instance when $\ALG$ leaves branch $L$ and enters branch $R$. Let $l_2$ be $\ADV_L$'s position at the end of the sub-instance, and let $r_2$ be the (unique) node in $R$ that represents $\ADV_R$'s position and $\ALG$'s new position after switching branches. We continue by playing a sub-instance of width $k-1$ rooted at $r_2$, and appending a path of zero-length edges to $l_2$ (while all other nodes in $L$ have no descendants in the next layers). We repeat this procedure many times, always playing a sub-instance of width $k-1$ in the branch where $\ALG$ is located, until it switches to the other branch. Figure \ref{fig:LB-adaptive} shows the constructed instance.

\begin{figure}[h]
\centering
\begin{tikzpicture}[->,>=stealth,auto]

\tikzset{
  state/.style={
    circle,
    draw,
    minimum size=4pt, 
    inner sep=0pt
  }
}

\node[state,label=left:$s$] (s) {};
\node[state,above right=1.5cm and 2.5cm of s,label=above left:$l_1$] (a1) {};
\node[below right=0.2cm and 0.9cm of a1,label=:$\cdots$] (dots1) {};
\node[state,right=2cm of a1,label=above right:$l_2$] (a2) {};
\node[state,right=2cm of a2,label=above left:$l_3$] (a3) {};
\node[below right=0.2cm and 0.9cm of a3,label=:$\cdots$] (dots1) {};
\node[state,right=2cm of a3,label=above right:$l_4$] (a4) {};
\node[state,right=2cm of a4,label=above:$l_5$] (a5) {};

\node[state,below right=1.5cm and 2.5cm of s,label=left:$r_1$] (b1) {};
\node[state,right=2cm of b1,label=below left:$r_2$] (b2) {};
\node[below right=0.2cm and 0.9cm of b2,label=:$\cdots$] (dots1) {};
\node[state,right=2cm of b2,label=below right:$r_3$] (b3) {};
\node[state,right=2cm of b3,label=below left:$r_4$] (b4) {};
\node[below right=0.2cm and 0.9cm of b4,label=:$\cdots$] (dots1) {};
\node[state,right=2cm of b4,label=below right:$r_5$] (b5) {};

\path (s) edge node[left] {$C_k$} (a1);
\path (s) edge node[left] {$C_k$} (b1);

\path (a2) edge node[above] {$0$} (a3);
\path (a4) edge node[above] {$0$} ++(2.0,0);

\path (b1) edge node[below] {$0$} (b2);
\path (b3) edge node[below] {$0$} (b4);

\draw ($(a1)+(-0.1,0.5)$) -- ($(a1)+(0.1,0.5)$) -- ($(a2)+(0.1,0.5)$) --
      ($(a2)+(0.1,-0.5)$) -- ($(a1)+(-0.1,-0.5)$) -- cycle
      node[pos=0.5,below] {};

\draw ($(a3)+(-0.1,0.5)$) -- ($(a3)+(0.1,0.5)$) -- ($(a4)+(0.1,0.5)$) --
($(a4)+(0.1,-0.5)$) -- ($(a3)+(-0.1,-0.5)$) -- cycle
node[pos=0.5,below] {};

\draw ($(b2)+(-0.1,0.5)$) -- ($(b2)+(0.1,0.5)$) -- ($(b3)+(0.1,0.5)$) --
($(b3)+(0.1,-0.5)$) -- ($(b2)+(-0.1,-0.5)$) -- cycle
node[pos=0.5,below] {};

\draw ($(b4)+(-0.1,0.5)$) -- ($(b4)+(0.1,0.5)$) -- ($(b5)+(0.1,0.5)$) --
($(b5)+(0.1,-0.5)$) -- ($(b4)+(-0.1,-0.5)$) -- cycle
node[pos=0.5,below] {};

\node[right=2cm of b5] (bdots) {$\cdots$};
\path (b5) edge[dashed] (bdots); 

\node[right=2cm of a5] (adots) {$\cdots$};
\path (a5) edge[dashed] (adots); 

\end{tikzpicture}
\caption{Lower bound against adaptive online adversaries.}
\label{fig:LB-adaptive}
\end{figure}
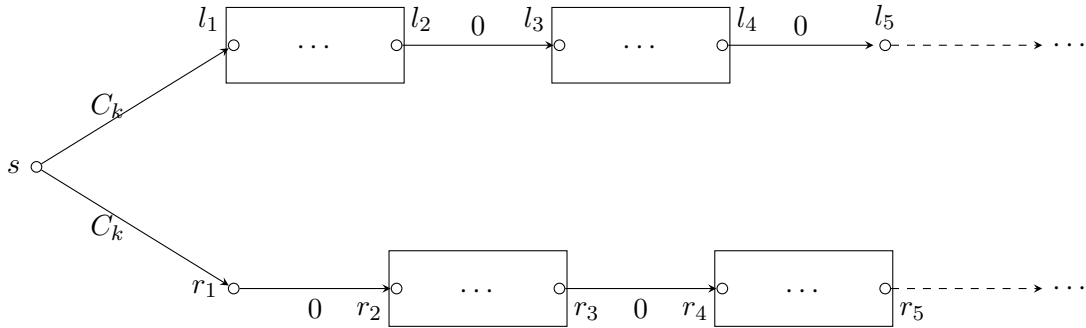

    By abusing notation, let $\ALG(x)$ be the cost incurred by $\ALG$ on the sub-instance rooted at node $x$, and let $\ADV(x)$ be the expected cost incurred by the adversary running in the sub-instance rooted at $x$. Recalling that $s$ denotes the source node in our top-level instance, we will prove that 
    \begin{equation}
    \label{eq:LB-adaptive-ineq}
    \ALG(s) \geq 2^{k-1} (\ADV(s) - C_k).
        \end{equation}

    The proof is by induction on $k$. The base case is when $k=1$, in which case the instance simply consists of a path and $\ALG(s) = \ADV(s)$, so \eqref{eq:LB-adaptive-ineq} holds. For the inductive step, note that 
    $\ALG$ pays cost $C_k$ to advance to $l_1$, cost $\ALG(l_1)$ during the sub-instance rooted at $l_1$, cost at least $2C_k$ to switch from $L$ to $R$, cost $\ALG(r_2)$ during the sub-instance rooted at $r_2$, and so on. Therefore, by the inductive hypothesis, the total cost incurred by $\ALG$ is at least
    \begin{align*}
    \ALG(s) &\geq C_k + \ALG(l_1) + 2 C_k + \ALG(r_2) + 2 C_k + \ALG(l_3) + 2 C_k + \ALG(r_4) + \dots \\
    &\geq \left(C_k + 2^{k-2} (\ADV(l_1) - C_{k-1})\right) + \left(2 C_k + 2^{k-2} (\ADV(r_2) - C_{k-1})\right) + \dots 
    \end{align*}
    We now choose $C_k$ so that $C_k \geq C_{k-1} \cdot 2^{k-2}$. Take, for example, $C_k = 2^{k^2}$. 

    Let $\ADV(L)$ and $\ADV(R)$ be the overall expected cost of the adversary on the sub-instances, conditioned on it choosing branch $L$ or $R$, respectively. We have 
    \begin{align*}
        \ADV(L) &= \ADV(l_1) + \ADV(l_3) + \dots \\
        \ADV(R) &= \ADV(r_2) + \ADV(r_4) + \dots\ ,
    \end{align*}
    so we obtain 
    \[
    \ALG(s) \geq 2^{k-2} (\ADV(L) + \ADV(R)). 
    \]

    The expected cost incurred by the adversary is 
    \[
    \ADV(s) = C_k + \frac{1}{2}(\ADV(L) + \ADV(R)),
    \]
    which concludes the proof of inequality \eqref{eq:LB-adaptive-ineq}.

    Finally, note that we can append as many sub-instances to our instance as desired. Moreover, at any time step the adversary is located at a given node in the last revealed layer with probability at least $2^{-(k-1)}$. In particular, it is located in the same position as $\ALG$ with probability at least $2^{-(k-1)}$. Since the instance is always extended with edges of positive length adjacent to $\ALG$'s position, the adversary pays cost at least $2^{-(k-1)} C_1$ per time step in expectation. Therefore, $\ADV(s)$ can be made arbitrarily large by stacking arbitrarily many sub-instances, so we can stop the instance after $C_k$ becomes 
    negligible with respect to $\ADV(s)$. By \eqref{eq:LB-adaptive-ineq}, this yields a lower bound of $2^{k-1}$ on the competitive ratio of $\ALG$ against the adaptive online adversary $\ADV$.\footnote{The reason why we obtain a lower bound of $2^{k-1}$ rather than $2^{k-2}$ (as in \cite{Fiat-competitiveLGT}) is that we make $C_k$ negligible with respect to $\ADV(s)$, rather than just ensuring that $C_k \leq \frac{1}{2}\ADV(s)$.}

\section{Reduction to the evolving tree game}
\label{sec:reduction}

In this section, we provide a reduction from layered graph traversal to the evolving tree game, proving Lemma~\ref{lemma:reduction-to-ETG}. The reduction is very similar to~\cite{randomizedLGT}.
\begin{lemma}
    Suppose there exists a $\rho$-competitive algorithm for width-$k$ layered graph traversal in which the graph is restricted to binary trees which have at most one edge of positive weight between consecutive layers. Then, there exists a $\rho$-competitive algorithm for width-$k$ layered graph traversal without these restrictions.
    
    \label{lem:trees-are-enough-in-LGT}
\end{lemma}
\begin{proof}
    \cite{Fiat-competitiveLGT} showed that it is enough to consider the problem restricted on trees. Note that any tree can be converted into a binary tree by inserting additional layers and edges of length zero, on the fly. Similarly, one can ensure that there is at most one edge of positive length between any two consecutive layers.
\end{proof}

\reductionToETG*
\begin{proof}
For a tree $T$, let $V(T)$ denote the vertices of $T$ and $\mathcal{L}(T)$ denote the leaves of $T$.
    Let $LT$ be an instance of layered graph traversal with maximum width $k$. By Lemma~\ref{lem:trees-are-enough-in-LGT}, we can assume that $LT$ is a rooted binary tree with at most one edge of positive weight between any two consecutive layers. To simplify our proof, we append an edge of length $0$ from the root $s$ of $LT$ to a new node $v$ located in a virtual layer $-1$. We denote by $L_i$ the vertices of $LT$ located in layer $i$, and we denote by $ST_i$ the Steiner tree on $LT$ with terminals $\{v\} \cup L_i$ (i.e.\ the minimal subtree of $LT$ which contains all paths from $v$ to a vertex in $L_i$). Note that $\mathcal{L}(ST_i) = L_i$.
    
    Suppose there exists an algorithm $\ALG^{ET}$ which is $\rho$-competitive for the evolving tree game on stemmed trees of maximum width $k$. We construct an algorithm $\ALG^{LGT}$ for traversing $LT$ by simulating 
    $\ALG^{ET}$ on an evolving tree $ET$ with root $r$. Our construction works in stages, and at the end of stage $i$ we will ensure that the following conditions hold:
    \begin{enumerate}
        \item There exists a homomorphism $f: V(ET) \rightarrow V(ST_i)$ such that 
        \begin{enumerate}
            \item The root $r$ of $ET$ is mapped to the node $v$ of $LT$. 
            \item Each leaf of $ET$ is mapped to a leaf of $ST_i$.
            \item Each edge of $ET$ is mapped to a path of the same length in $ST_i$.
        \end{enumerate}
        \item If $\ALG^{ET}$ is in a leaf $l$, $\ALG^{LGT}$ is in $f(l) \in L_i$.
    \end{enumerate}

    At stage $0$, we initialise $ET$ to consist of an edge of length $0$ connecting the root $r$ to its child $u$, and we set $f(r) = v$ and $f(u) = s$. Also, $\ALG^{ET}$ is located in $u$ and $\ALG^{LGT}$ is located in $s$, so the required conditions are satisfied.

    In the beginning of stage $i\geq 1$, layer $i$ is revealed and $ST_i$ can be computed. The evolving tree $ET$ is currently homomorphic to $ST_{i-1}$, and we need to transform it so that it is homomorphic to $ST_i$. As long as there exists a leaf $l \in \mathcal{L}(ET)$ for which $f(l) \notin V(ST_i)$, we delete $l$. Next, we fork all $l \in \mathcal{L}(ET)$ for which $f(l)$ has two neighbours $\{x, y\} \in L_i$; we extend $f$ by mapping the newly created vertices in $ET$ to $\{x, y\}$. For the remaining leaves $l \in \mathcal{L}(ET)$ for which $f(l)$ has exactly one neighbour $x \in L_i$, we change $f$ so that $f(l) := x$. Finally, if there exists an edge $\{x, y\}$ of positive length $p$ between $x \in L_{i-1}$ and $y \in L_i$, we grow $f^{-1}(y)$ by $p$. At the end of the stage, $\ALG^{LGT}$ moves to $f(l^*) \in L_i$, where $l^*$ is the position of $\ALG^{ET}$. 
    Note that the evolving tree $ET$ has at most $k$ leaves at any given time, assuming that $LT$ has at most $k$ vertices in each layer. Let $N$ be the total number of layers. Since layer $N$ contains only one vertex, namely the target $t$, the evolving tree $ET$ contains only one leaf $f^{-1}(t)$. Since $f$ is a homomorphism, the optimum (and only) root-to-leaf path in $ET$ has the same length as the path from the source $s$ to the target $t$ in $LT$. Denote the length of this path by $\dist(s,t)$. By the assumption that $\ALG^{ET}$ is $\rho$-competitive, we have $\cost(\ALG^{ET}) \leq \rho \cdot \dist(s, t)$. By the triangle inequality, it is easy to see that $\cost(\ALG^{LGT}) \leq \cost(\ALG^{ET})$, which concludes the proof. 
\end{proof}

\end{document}